\documentclass[runningheads]{llncs}
\usepackage{graphicx}
\usepackage{amsmath, amssymb, amsthm}
\usepackage{thmtools, thm-restate}
\usepackage{mathtools}
\usepackage{stmaryrd}
\usepackage{todonotes}
\usepackage{url}
\usepackage{xfrac}
\usepackage{csquotes}
\usepackage{xspace}
\usepackage{listings}
\usepackage{wrapfig}
\usepackage{tikz-cd}
\usepackage[capitalize]{cleveref}
\usepackage{lineno}
\usepackage{apptools}

\crefname{lstinputlisting}{program}{programs}
\Crefname{lstinputlisting}{Program}{Programs}

\newcommand{\codify}[1]{\textup{\texttt{#1}}}
\newcommand{\while}{\codify{while}\xspace}
\newcommand{\ASSIGN}[2]{{#1} \coloneqq {#2}}

\newcommand{\WHILEDO}[2]{\codify{while}\:(#1)\:\{#2\}}
\newcommand{\pgcl}{\codify{pGCL}\xspace}
\newcommand{\pskip}{\codify{skip}\xspace}

\newcommand{\pif}{\codify{if}}
\newcommand{\ITE}[3]{\codify{if}\:(#1)\:\{#2\} \: \codify{else} \: \{#3\}}
\newcommand{\PCHOICE}[3]{\{#1\} \: [#2] \: \{#3\}}
\newcommand{\compose}[2]{{#1}{\fatsemi} \: {#2}}

\newcommand{\pelse}{\codify{else}}
\newcommand{\sem}[1]{\ensuremath{\left\llbracket #1 \right\rrbracket}}
\newcommand{\semapp}[2]{\ensuremath{\left\llbracket #1 \right\rrbracket \left( #2 \right) }}
\newcommand{\constrain}[2]{\restrict{#1}{#2}}
\newcommand{\abs}[1]{\ensuremath{\vert{#1}\vert}}

\newcommand{\evalstate}[2]{\eval{#2}{#1}}
\newcommand{\restrict}[2]{\ensuremath{\left\langle #1 \right\rangle_{#2}}}
\definecolor{webgreen}{rgb}{0,.5,0}
\newcommand{\green}[1]{\textcolor{webgreen}{#1}}
\newcommand{\annotate}[1]{\green{\!\!{\fatslash}\!\!{\fatslash}~~\vphantom{G'} {#1}}}

\newcommand{\phifct}[2]{\Phi_{#1, #2}}
\newcommand{\phifctapp}[3]{\phifct{#1}{#2} \left( #3\right)}
\newcommand{\phifctpowerapp}[4]{\phifct{#1}{#2}^{#3} \left( #4\right)}

\newcommand{\phifctres}[3]{\Phi_{#1, #2, #3}}

\newcommand{\extract}[2]{\ensuremath{\left[{#1}\right]_{#2}}}
\newcommand{\mass}[1]{\ensuremath{\left|{#1}\right|}}
\newcommand{\bvec}[1]{\ensuremath{\mathbf{#1}}}

\newcommand{\ppto}{~{}\dashrightarrow{}~}
\newcommand{\tto}{~{}\to{}~}

\newcommand{\ccoloneqq}{~{}\coloneqq{}~}
\newcommand{\eval}[2]{\textsf{eval}_{#2}(#1)}

\newcommand{\pgf}{\textnormal{\textsf{PGF}}\xspace}
\newcommand{\ps}{\textnormal{\textsf{FPS}}\xspace}
\newcommand{\mon}[1]{\ensuremath{\textnormal{Mon}\left(#1\right)\xspace}}
\newcommand{\var}{\textnormal{Var}\xspace}
\newcommand{\N}{\ensuremath{\mathbb{N}}\xspace}

\newcommand{\PosRealsInf}{\ensuremath{\mathbb{R}_{\geq 0}^\infty}\xspace}
\newcommand{\PosReals}{\ensuremath{\mathbb{R}_{\geq 0}}\xspace}
\newcommand{\lfp}{\ensuremath{\textnormal{\textsf{lfp}}~}\xspace}

\newcommand{\distr}{\mathcal{D}}
\newcommand{\subdistr}{\mathcal{D}_{\leq}}

\newcommand{\ppreceq}{~{}\preceq{}~}

\newcommand{\eeq}{~{}={}~}

\newcommand{\lleq}{~{}\leq{}~}

\newcommand{\ssqsubseteq}{~{}\sqsubseteq{}~}

\newcommand{\iimplies}{~{}\implies{}~}
\newcommand{\pplus}{~{}+{}~}
\newcommand{\mminus}{~{}-{}~}

\newcommand{\mmapsto}{~{}\mapsto{}~}

\newcommand{\bigmid}{~\big|~}

\newcommand{\qqiff}{\qquad\textnormal{iff}\qquad}

\newcommand{\qqand}{\qquad\textnormal{and}\qquad}

\newcommand{\qimplies}{\quad\textnormal{implies}\quad}

\newcommand{\qqimplies}{\qquad\textnormal{implies}\qquad}

\newcommand{\program}[3]{\begin{figure*}[#3] \begin{center}\lstinputlisting[mathescape=true, caption=#2, captionpos=b, label={prog:#1}, frame=l, framerule=1.5pt, framesep=2pt, backgroundcolor = \color{lightgray!50!white}]{Programs/#1}\end{center}\end{figure*}}

\newcommand{\Pref}[1]{Program \ref{#1}}

\begin{document}
\title{Generating Functions for Probabilistic Programs\thanks{This research was funded by the ERC AdG project FRAPPANT (787914) and the DFG RTG 2236 UnRAVeL.}
}

%
%
\author{Lutz Klinkenberg\inst{1}\orcidID{0000-0002-3812-0572} \and
Kevin Batz\inst{1}\orcidID{0000-0001-8705-2564} \and
Benjamin Lucien Kaminski \inst{1,2}\orcidID{0000-0001-5185-2324} \and
Joost-Pieter Katoen \inst{1} \orcidID{0000-0002-6143-1926} \and
Joshua Moerman \inst{1} \orcidID{0000-0001-9819-8374} \and
Tobias Winkler \inst{1} \orcidID{0000-0003-1084-6408}}
\authorrunning{L. Klinkenberg et al.}
%
\institute{RWTH Aachen University, 52062 Aachen, Germany \and
University College London, United Kingdom
\\
\email{\{lutz.klinkenberg, kevin.batz, benjamin.kaminski, katoen, joshua, tobias.winkler\}@cs.rwth-aachen.de}}
\maketitle              
\begin{abstract}
This paper investigates the usage of generating functions (GFs) encoding measures over the program variables for reasoning about discrete probabilistic programs.
To that end, we define a denotational GF-transformer semantics for probabilistic while-programs, and show that it instantiates Kozen's seminal distribution transformer semantics.
We then study the effective usage of GFs for program analysis.
We show that finitely expressible GFs enable checking super-invariants by means of computer algebra tools, and that they can be used to determine termination probabilities.
The paper concludes by characterizing a class of --- possibly infinite-state --- programs whose semantics is a rational GF encoding a discrete phase-type distribution.

\end{abstract}

\section{Introduction}\label{sec:intro}

Probabilistic programs are sequential programs for which coin flipping is a first-class citizen.
They are used e.g.\ to represent randomized algorithms, probabilistic graphical models such as Bayesian networks, cognitive models, or security protocols.
Although probabilistic programs are typically rather small, their analysis is intricate. 
For instance, approximating expected values of program variables at program termination is as hard as the universal halting problem~\cite{DBLP:journals/acta/KaminskiKM19}.
Determining higher moments such as variances is even harder.
Deductive program verification techniques based on a quantitative version of weakest preconditions~\cite{DBLP:series/mcs/McIverM05} enable to reason about the outcomes of probabilistic programs, such as what is the probability that a program variable equals a certain value. 
Dedicated analysis techniques have been developed to e.g., determine tail bounds~\cite{DBLP:conf/tacas/BouissouGPCS16}, decide almost-sure termination~\cite{DBLP:journals/pacmpl/McIverMKK18,DBLP:journals/toplas/ChatterjeeFNH18}, or to compare programs~\cite{DBLP:conf/popl/BartheGHS17}.

This paper aims at exploiting the well-tried potential of \emph{probability generating functions}~(PGFs~\cite{JKK1993}) for the analysis of probabilistic programs. 
In our setting, PGFs are power series representations --- generating functions --- encoding \emph{discrete} probability mass functions of joint distributions over program variables.
PGF representations \mbox{--- in} particular if finite --- enable a simple extraction of important information from the encoded distributions such as expected values, higher moments, termination probabilities or stochastic independence of program variables.

To enable the usage of PGFs for program analysis, we define a denotational semantics of a simple probabilistic while-language akin to probabilistic GCL~\cite{DBLP:series/mcs/McIverM05}.
Our semantics is defined in a \emph{forward} manner: given an input distribution over program variables as a PGF, it yields a PGF representing the resulting subdistribution.
The \enquote{missing} probability mass represents the probability of non-termination.
More accurately, our denotational semantics transforms \emph{formal power series} (FPS).
Those form a richer class than PGFs, which allows for overapproximations of probability distributions.
While-loops are given semantics as least fixed points of FPS transformers.
It is shown that our semantics is in fact an instantiation of Kozen's seminal distribution-transformer semantics~\cite{Kozen79}.

The semantics provides a sound basis for program analysis using PGFs.
Using Park's Lemma, we obtain a simple technique to prove whether a given FPS overapproximates a program's semantics i.e., whether an FPS is a so-called super-invariant.
Such upper bounds can be quite useful: for almost-surely terminating programs, such bounds can provide \emph{exact} program semantics, whereas, if the mass of an overapproximation is strictly less than one, the program is \emph{provably non-almost-surely terminating}.
This result is illustrated on a non-trivial random walk and on examples illustrating that checking whether an FPS is a super-invariant can be \emph{automated} using computer algebra tools.

In addition, we characterize a class of --- possibly infinite-state --- programs whose PGF semantics is a rational function.
These \emph{homogeneous bounded programs} (HB programs) are characterized by loops in which each unbounded variable has no effect on the loop guard and is in each loop iteration incremented by a quantity independent of its own value.
Operationally speaking, HB programs can be considered as finite-state Markov chains with rewards in which rewards can grow unboundedly large. 
It is shown that the rational PGF of any program that is equivalent to an almost-surely terminating HB program represents a multi-variate discrete phase-type distribution~\cite{Neuts81}. 
We illustrate this result by obtaining a closed-form characterization for the well-studied infinite-state dueling cowboys example~\cite{DBLP:series/mcs/McIverM05}.

\paragraph{Related work.}
Semantics of probabilistic programs is a well-studied topic.
This includes the seminal works by Kozen~\cite{Kozen79} and McIver and Morgan~\cite{DBLP:series/mcs/McIverM05}.
Other related semantics of discrete probabilistic while-programs are e.g., given in several other articles like~\cite{DBLP:series/mcs/McIverM05,DBLP:conf/aplas/PierroW13,DBLP:journals/pe/GretzKM14,DBLP:conf/aaai/NoriHRS14,DBLP:conf/esop/BichselGV18}.
PGFs have recent scant attention in the analysis of probabilistic programs. 
A notable exception is~\cite{DBLP:conf/icalp/Boreale15} in which generating functions of finite Markov chains are obtained by Pad\'{e} approximation. 
Computer algebra systems have been used to transform probabilistic programs~\cite{DBLP:conf/padl/CaretteS16}, and more recently in the automated generation of moment-based loop invariants~\cite{DBLP:conf/atva/BartocciKS19}.

\paragraph{Organization of this paper.}
After recapping FPSs and PGFs in Sections \ref{sec:prelim}--\ref{sec:gf}, we define our FPS transformer semantics in \Cref{sec:semantics}, discuss some elementary properties and show it instantiates Kozen's distribution transformer semantics~\cite{Kozen79}. 
\Cref{sec:overapprox} presents our approach for verifying upper bounds to loop invariants and illustrates this by various non-trivial examples.
In addition, it characterizes programs that are representable as finite-state Markov chains equipped with rewards and presents the relation to discrete phase-type distributions.
\Cref{sec:conclusion} concludes the paper.
All proofs can be found in the appendix.

\section{Formal Power Series}\label{sec:prelim}

Our goal is to make the potential of probability generating functions available to the formal verification of probabilistic programs.
The programs we consider will, without loss of generality, operate on a fixed set of $k$ program variables.
The valuations of those variables range over $\N$.
A \emph{program state} $\sigma$ is hence a vector in $\N^k$.
We denote the state $(0,\ldots,0)$ by $\vec{0}$.

A prerequisite for understanding probability generating functions are (multivariate) \emph{formal power series} --- \emph{a special way of representing a potentially infinite $k$-dimensional array}.
For $k{=}1$, this amounts to representing a \emph{sequence}.
\begin{definition}[Formal Power Series]
	Let $\bvec{X} = X_1,\, \ldots,\, X_k$ be a fixed sequence of $k$ distinct formal indeterminates.
	For a state $\sigma = (\sigma_1,\, \ldots,\, \sigma_k) \in \N^k$, let $\bvec{X}^\sigma$ abbreviate the formal multiplication $X_1^{\sigma_1} \mathbin{\cdots} X_k^{\sigma_k}$.
	The latter object is called a \emph{monomial} and we denote the set of all monomials over $\bvec{X}$ by $\mon{\bvec{X}}$.
	A (multivariate) \emph{formal power series (FPS)} is a formal sum
	\[
		F \eeq \sum_{\sigma  \in \N^k} \extract{\sigma}{F} \cdot \bvec{X}^\sigma~, \qquad \textnormal{where} \qquad \extract{\: \cdot \:}{F}\colon \quad \N^k \to \PosRealsInf~,
	\]
	where $\PosRealsInf$ denotes the extended positive real line.
	We denote the set of all FPSs by $\ps$.
	Let $F, G \in \ps$.
	If $\extract{\sigma}{F} < \infty$ for all $\sigma \in \N^k$, we denote this fact by $F \ll \infty$.
	The \emph{addition}~$F + G$ and scaling~\mbox{$r \cdot F$} by a scalar $r \in \PosRealsInf$ is \mbox{defined coefficient-wise by}%
	\[
		F + G \eeq \sum_{\sigma  \in \N^k} \bigl(\extract{\sigma}{F} + \extract{\sigma}{G} \bigr ) \cdot \bvec{X}^\sigma 
		\qqand
		r \cdot F \eeq \sum_{\sigma  \in \N^k} r \cdot \extract{\sigma}{F} \cdot \bvec{X}^\sigma ~.
	\]%
	For states $\sigma = (\sigma_1,\, \ldots,\, \sigma_k)$ and $\tau = (\tau_1,\, \ldots,\, \tau_k)$, we define $\sigma + \tau = (\sigma_1 + \tau_1,\, \ldots,\, \sigma_k + \tau_k)$.
	The \emph{multiplication}~$F \cdot G$ is given as their Cauchy product (or discrete convolution)
	\[
		F \cdot G \eeq \sum_{\sigma,\tau  \in \N^k} \extract{\sigma}{F} \cdot \extract{\tau}{G} \cdot \bvec{X}^{\sigma + \tau}~.
	\]%
\end{definition}%
\noindent{}Drawing coefficients from the extended reals enables us to define a \emph{complete lattice} on FPSs in \Cref{sec:semantics}.
Our analyses in \Cref{sec:overapprox} will, however, only consider FPSs with $F \ll \infty$.

\section{Generating Functions}
\label{sec:gf}

\begin{quote}
	A generating function is a device somewhat similar to a bag.
	Instead of carrying many little objects detachedly, which could be embarrassing, we put them all in a bag, and then we have only one object to carry, the bag.\\[-1.00\baselineskip]
	\begin{flushright}
	--- George P\'{o}lya~\cite{polya1954mathematics}
\end{flushright}
\end{quote}%
Formal power series pose merely a particular way of encoding an infinite $k$-dimensional array as yet another infinitary object, but we still carry all objects forming the array (the coefficients of the FPS) detachedly and there seems to be no advantage in this particular encoding.
It even seems more bulky.
We will now, however, see that this bulky encoding can be turned into a one-object bag carrying all our objects: the \emph{generating function}.%
\begin{definition}[Generating Functions]\label{def:generating-functions}
The \emph{generating function} of a formal power series $F = \sum_{\sigma  \in \N^k} [\sigma]_F \cdot \bvec{X}^\sigma \in \ps$ with $F \ll \infty$ is defined as the \emph{partial} function
	\[
		f \colon \quad [0,\, 1]^k \ppto \PosReals, \quad (x_1,\, \ldots,\, x_k) \mmapsto \quad\sum_{\mathclap{\sigma = (\sigma_1,\ldots,\sigma_k) \in \N^k}}~ \extract{\sigma}{F} \cdot x_1^{\sigma_1} \cdots x_k^{\sigma_k}~.
	\]%
\end{definition}%
\noindent{}%
In other words: in order to turn an FPS into its generating function, we merely treat every \emph{formal} indeterminate $X_i$ as an \emph{\enquote{actual}} indeterminate $x_i$, and the formal multiplications and the formal sum also as \emph{\enquote{actual}} ones.
The generating function $f$ of $F$ is \emph{uniquely determined} by $F$ as we require all coefficients of $F$ to be non-negative, and so the ordering of the summands is irrelevant:
For a given point $\vec{x} \in [0,\, 1]^k$, the sum defining $f(\vec{x})$ either converges \emph{absolutely} to some positive real or diverges absolutely to $\infty$.
In the latter case, $f$ is undefined at $\vec{x}$ and hence $f$ may indeed be partial.

Since generating functions stem from formal power series, they are infinitely often differentiable at $\vec{0} = (0,\ldots,0)$.
Because of that, we can recover $F$ from $f$ as the (multivariate) Taylor expansion of $f$ at~$\vec{0}$.%
\begin{definition}[Multivariate Derivatives and Taylor Expansions]\label{def:mult_taylor}
For $\sigma = (\sigma_1,\ldots,\sigma_k) \in \N^k$, we write $f^{(\sigma)}$ for the function $f$ differentiated $\sigma_1$ times in $x_1$, $\sigma_2$ times in $x_2$, and so on.
If $f$ is infinitely often differentiable at $\vec{0}$, then the \emph{Taylor expansion of $f$ at $\vec{0}$} is given by
\[
	\sum_{\sigma \in \N^k}~ \frac{f^{(\sigma)}\left(\,\vec{0}\,\right)}{\sigma_1!  \mathbin{\cdots}  \sigma_k!} \cdot x_1^{\sigma_1} \cdots x_k^{\sigma_k}~.
\]
\end{definition}%
\noindent{}%
If we replace every indeterminate $x_i$ by the \emph{formal} indeterminate $X_i$ in the Taylor expansion of generating function $f$ of $F$, then we obtain the formal power series~$F$.
It is in precisely that sense, that $f$ \emph{generates} $F$.%
\begin{example}[Formal Power Series and Generating Functions]\label{ex:ps-gf}
	Consider the infinite \mbox{(1-dimensional)} sequence $\sfrac{1}{2},\, \sfrac{1}{4},\, \sfrac{1}{8},\, \sfrac{1}{16},\, \ldots$.
	Its (univariate) FPS --- the entity carrying all coefficients detachedly --- is given as
	\[
		\frac{1}{2} + \frac{1}{4} X + \frac{1}{8}X^2 + \frac{1}{16}X^3 + \frac{1}{32}X^4 + \frac{1}{64}X^5+ \frac{1}{128}X^6 + \frac{1}{256}X^7 + \ldots~. \tag{$\dagger$}
	\]
	On the other hand, its generating function --- the bag --- is given concisely by
	\[
		\frac{1}{2 - x}~. \tag{$\flat$}
	\]
	Figuratively speaking, $(\dagger)$ is itself the infinite sequence $a_n \coloneqq\tfrac{1}{2^n}$, whereas $(\flat)$~is a bag with the label \mbox{\enquote{infinite sequence $a_n\coloneqq\tfrac{1}{2^n}$}}.
	The fact that $(\dagger)$ generates $(\flat)$, follows from the Taylor expansion of $\tfrac{1}{2 - x}$ at~$0$ being $\frac{1}{2} + \frac{1}{4}x + \frac{1}{8}x^2 + \ldots$. \hfill$\triangle$
\end{example}%
\noindent{}%
The potential of generating functions is that manipulations to the functions ---~i.e.\ to the concise representations~--- are in a one-to-one correspondence to the associated manipulations to FPSs~\cite{graham1994concrete}.
For instance, if $f(x)$ is the generating function of $F$ encoding the sequence $a_1,\, a_2,\, a_3,\, \ldots$, then the function $f(x) \cdot x$ is the generating function of $F \cdot X$ which encodes the sequence $0,\, a_1,\, a_2,\, a_3,\, \ldots$

As another example for correspondence between operations on FPSs and generating functions, if $f(x)$ and $g(x)$ are the generating functions of $F$ and $G$, respectively, then $f(x) + g(x)$ is the generating function of $F + G$.
\begin{example}[Manipulating to Generating Functions]
	Revisiting Example \ref{ex:ps-gf}, if we multiply $\tfrac{1}{2 - x}$ by $x$, we change the label on our bag from \enquote{infinite sequence $a_n\coloneqq\tfrac{1}{2^n}$} to \mbox{\enquote{a $0$ followed by an infinite sequence $a_{n+1}\coloneqq \tfrac{1}{2^n}$}} and --- just by changing the label --- the bag will now contain what it says on its label. 
	Indeed, the Taylor expansion of~$\tfrac{x}{2 - x}$~at~$0$ is $0 + \frac{1}{2}x + \frac{1}{4}x^2 + \frac{1}{8}x^3 + \frac{1}{16}x^4 + \ldots$ encoding the sequence $0,\, \sfrac{1}{2},\, \sfrac{1}{4},\, \sfrac{1}{8},\, \sfrac{1}{16},\, \ldots$
	\hfill$\triangle$
\end{example}%
\noindent{}%
Due to the close correspondence of FPSs and generating functions~\cite{graham1994concrete}, we use both concepts interchangeably, as is common in most mathematical literature.
We mostly use FPSs for definitions and semantics, and generating functions in calculations and examples.

\medskip\noindent\textbf{Probability Generating Functions.}
We now use formal power series to represent probability distributions.%
\begin{definition}[Probability Subdistribution]
	A \emph{probability subdistribution} (or simply subdistribution) over $\N^k$ is a function
	\[
	\mu \colon\quad \N^k \tto [0,1],\qquad \text{such that}\qquad \mass{\mu} \eeq \sum_{\sigma \in \N^k}\mu(\sigma) \leq 1~. 
	\]
	We call $\mass{\mu}$ the \emph{mass of $\mu$}.
	We say that $\mu$ is a \emph{(full) distribution} if $\mass{\mu} = 1$, and a \emph{proper subdistribution} if $\mass{\mu} < 1$.
	The set of all subdistributions on $\N^k$ is denoted by $\subdistr(\N^k)$ and the set of all full distributions by $\distr(\N^k)$.
\end{definition}%
\noindent{}%
We need subdistributions for capturing non-termination.
The \enquote{missing} probability mass $1 - \mass{\mu}$ precisely models the probability of non-termination.

The generating function of a (sub-)distribution is called a \emph{probability generating function}. Many properties of a distribution $\mu$ can be read off from its generating function $G_\mu$ in a simple way. We demonstrate how to extract a few common properties in the following example.
\begin{example}[Geometric Distribution PGF]\label{ex:geo-gf}
	Recall \cref{ex:ps-gf}. The presented formal power series encodes a \emph{geometric distribution} $\mu_{\mathit{geo}}$ with parameter $\sfrac{1}{2}$ of a single variable $X$.
   The fact that $\mu_{\mathit{geo}}$ is a proper probability distribution, for instance, can easily be verified computing $G_\mathit{geo}(1) = \tfrac{1}{2 - 1} = 1$. The expected value of $X$ is given by $G_{\mathit{geo}}'(1) = \tfrac{1}{(2-1)^2} = 1$.\hfill$\triangle$
\end{example}

\medskip\noindent\textbf{Extracting Common Properties.}
	Important information about probability distributions is, for instance, the first and higher moments.
	In general, the $k^\textnormal{th}$ factorial moment of variable $X_i$ can be extracted from a PGF by computing $\tfrac{\partial^k G}{\partial X_i^k}(1,\ldots,1)$.\footnote{In general, one must take the limit $X_i \to 1$ from below.} This includes the mass $\mass{G}$ as the $0^\textnormal{th}$ moment. The marginal distribution of variable $X_i$ can simply be extracted from $G$ by $G(1,\ldots, X_i,\ldots,1)$.
	We also note that PGFs can treat \emph{stochastic independence}. 
	For instance, for a bivariate PGF $H$ we can check for stochastic independence of the variables $X$ and~$Y$ by checking whether $H(X,Y) = H(X,1) \cdot H(1,Y)$.

\section{FPS Semantics for \pgcl}\label{sec:semantics}

In this section, we give denotational semantics to probabilistic programs in terms of FPS transformers and establish some elementary properties useful for program analysis.
We begin by endowing FPSs and PGFs with an order structure:%
\begin{definition}[Order on \ps]
	For all $F, G \in \ps$, let
	\[
		F \ppreceq G
		\qqiff
		\forall\, \sigma \in \N^k \colon\quad \extract{\sigma} G \lleq \extract{\sigma} F ~.
    \]%
\end{definition}%
\begin{restatable}[Completeness of $\preceq$ on \ps]{lemma}{cocpos}
	\label{lem:completeness_of_pos}
    $(\ps,\, {\preceq})$ is a complete latttice.
\end{restatable}%

\subsection{FPS Transformer Semantics}\label{ssec:semantics}

Recall that we assume programs to range over exactly $k$ variables with valuations in $\N^k$.
Our program syntax is similar to Kozen~\cite{Kozen79} and McIver \& Morgan~\cite{DBLP:series/mcs/McIverM05}.%
\begin{definition}[Syntax of \pgcl{}~\cite{Kozen79,DBLP:series/mcs/McIverM05}]
	A program $P$ in \emph{probabilistic Guarded Command Language} (\pgcl{}) adheres to the grammar%
	\begin{align*}
	P \Coloneqq~& \pskip{} \bigmid 
	\codify{x}_{i} := E \bigmid 
	P;P
	 \bigmid\{P\}\ [p]\ \{P\}\\
	 ~&\bigmid\pif{}(B)\ \{P\}\ \pelse{}\ \{P\}
	 \bigmid\while{}\:(B)\ \{P\}~,
	\end{align*}%
	where $\codify{x}_i \in \{\codify{x}_1, \ldots, \codify{x}_k\}$ is a program variable, $E$ is an arithmetic expression over program variables, $p \in [0,1]$ is a probability, and $B$ is a predicate (called \emph{guard}) over program variables.
\end{definition}%
\noindent{}%
The FPS semantics of \pgcl{} will be defined in a forward denotational style, where the program variables $\codify{x}_1, \ldots, \codify{x}_k$ correspond to the formal indeterminates $X_1, \ldots, X_k$ of FPSs.

For handling assignments, $\codify{if}$-conditionals and $\codify{while}$-loops, we need some auxiliary functions on FPSs:
For an arithmetic expression $E$ over program variables, we denote by $\eval{E}{\sigma}$ the evaluation of $E$ in program state $\sigma$.
For a predicate~$B \subseteq \N^k$ and FPS $F$, we define the \emph{restriction of $F$ to $B$} by%
\[
	\restrict{F}{B} \ccoloneqq \sum_{\sigma \in B} \extract{\sigma}{F} \cdot \bvec{X}^\sigma~,
\]%
i.e.\ $\restrict{F}{B}$ is the FPS obtained from $F$ by setting all coefficients $\extract{\sigma}{F}$ where $\sigma \not\in B$ to $0$. 
Using these prerequisites, our FPS transformer semantics is given as follows:%
\begin{definition}[FPS Semantics of \pgcl{}]\label{def:semantics}
	The semantics $\sem{P} \colon \ps \to \ps$ of a loop-free \pgcl program $P$ is given according to the upper part of \textnormal{\Cref{tab:semantics}}.
	
	The \emph{unfolding operator} $\Phi_{B,P}$ for the loop $\WHILEDO{B}{P}$ is defined by%
	\[
	\Phi_{B,P} \colon \quad (\ps \to \ps) \to (\ps \to \ps), \quad \psi \mmapsto \lambda F\boldsymbol{.}~ \langle F\rangle_{\neg B} \pplus \psi \Bigl( \sem{P}\bigl(\langle F \rangle_B \bigr) \Bigr).
	\]%
	The partial order  $(\ps,\, {\preceq})$ extends to a partial order $\bigl(\ps \to \ps,\, {\sqsubseteq}\bigr)$ on FPS transformers by a point-wise lifting of $\preceq$.
	The least element of this partial order is the transformer~$\boldsymbol{0} = \lambda F\boldsymbol{.}~0$ mapping any FPS $F$ to the zero series.
	The semantics of $\WHILEDO{B}{P}$ is then given by the least fixed point (with respect to $\sqsubseteq$) of its unfolding operator, i.e.%
	\[
	\sem{\WHILEDO{B}{P}} \eeq \lfp \Phi_{B,P} ~.
	\]%
\end{definition}%
\begin{table*}[b]
	\caption{FPS transformer semantics of \pgcl programs.}
	\label{tab:semantics}
	\centering
	\def\arraystretch{1.25}%
	\begin{tabular*}{.88\textwidth}{l@{\qquad}l}
		$P$ & $\sem{P}(F)$\\[.25em]
		\hline
		\hline
		\pskip{} & $F$ \\
		$\codify{x}_i \coloneqq E$ & $\sum_{\sigma \in \N^k}\mu_\sigma X_{1}^{\sigma_1}\cdots X_{i}^{\eval{E}{\sigma}}\cdots X_{k}^{\sigma_k}$\\
		$\PCHOICE{P_1}{p}{P_2}$ & $p \cdot \sem{P_1}(F) \pplus (1-p) \cdot \sem{P_2}(F)$\\
		$\ITE{B}{P_1}{P_2}$ & $\sem{P_1}\bigl(\langle F\rangle_B \bigr) \pplus \sem{P_2} \bigl( \langle F \rangle_{\neg B} \bigr)$\\
		$\compose{P_1}{P_2}$ & $\sem{P_2}\bigl( \sem{P_1}(F) \bigl)$\\[.5em]
		\hline\\[-1em]
		$\while{}(B)\{P\}$ &$\bigl(\lfp\, \Phi_{B,P}\bigr)(F)$~, \quad{}for \\
		& \qquad $\Phi_{B,P}(\psi) \eeq \lambda F\boldsymbol{.}~ \langle F\rangle_{\neg B} \pplus \psi \Bigl( \sem{P}\bigl(\langle F \rangle_B \bigr) \Bigr)$
	\end{tabular*}%
\end{table*}%
\pagebreak{}%
\begin{example}
	\marginpar{\parbox{\linewidth}{
		\abovedisplayskip=2.6\baselineskip%
				\belowdisplayskip=-1em%
				\begin{align*}
					& \annotate{G} \\
					& P' \\
					& \annotate{G'} \\
				\end{align*}%
				\normalsize%
		}} %
    Consider the program $P = \compose{\PCHOICE{\ASSIGN{\codify{x}}{0}}{\sfrac{1}{2}}{\ASSIGN{\codify{x}}{1}}}{\ASSIGN{\codify{c}}{c+1}}$ and the input PGF $G = 1$, which denotes a point mass on state $\sigma = \vec{0}$.
    	Using the annotation style shown in the left margin, denoting that $\semapp{P'}{G} = G'$, we calculate $\semapp{P}{G}$ as follows:
	\begin{align*}
		& \annotate{1} \\
		& \compose{\!\PCHOICE{\ASSIGN{\codify{x}}{0}}{\sfrac{1}{2}}{\ASSIGN{\codify{x}}{1}}}{} \\
		& \annotate{\tfrac{1}{2} + \tfrac{X}{2}} \\
		&\ASSIGN{\codify{c}}{\codify{c}+1}\\
		& \annotate{\tfrac{C}{2} + \tfrac{CX}{2}}
	\end{align*}%
	As for the semantics of $\ASSIGN{c}{c + 1}$, see \cref{tab:ops}. \hfill$\triangle$
\end{example}%
\begin{table}[b]
	\caption{Common assignments and their effects on the input PGF $F(X,Y)$.}
	\label{tab:ops}
	\centering
	\def\arraystretch{1.25}%
	\begin{tabular*}{0.5\textwidth}{l@{\qquad}l}
		$P$ & $\sem{P}(F)$\\[.25em]
		\hline
		\hline
		$\ASSIGN{\codify{x}}{\codify{x + } k}$ & $X^k \cdot F(X,Y)$ \\
		$\ASSIGN{\codify{x}}{k \cdot \codify{x}}$ & $F(X^k, Y)$ \\
		$\ASSIGN{\codify{x}}{\codify{x + y}}$ & $F(X, XY)$ \\
	\end{tabular*}
\end{table}%
\noindent{}%
Before we study how our FPS transformers behave on PGFs in particular, we now first argue that our FPS semantics is well-defined.
While evident for loop-free programs, we appeal to the Kleene Fixed Point Theorem for loops~\cite{DBLP:journals/ipl/LassezNS82}, which requires $\omega$-continuous functions.%
\begin{restatable}[$\boldsymbol{\omega}$-continuity of \pgcl Semantics]{theorem}{cont}\label{thm:cont}
	The semantic functional~$\sem{\: \cdot \:}$ is $\omega$-continuous, i.e.\ for all programs $P \in \pgcl$ and all increasing $\omega $-chains \mbox{$F_1 \preceq F_2 \preceq \ldots$  in $\ps$},%
	\[
		\semapp{P}{\,\sup_{n \in \N} F_n} \eeq \sup_{n \in \N}~ \semapp{P}{F_n}~.
	\]
\end{restatable}%
\begin{restatable}[Well-definedness of FPS Semantics]{theorem}{welldef}\label{thm:well_definedness}
	The semantics functional $\sem{\: \cdot \:}$ is well-defined, i.e.\ the semantics of any loop $\WHILEDO{B}{P}$ exists uniquely and can \mbox{be written as}%
	\[
		\sem{\WHILEDO{B}{P}} \eeq \lfp \Phi_{B,P} \eeq \sup_{n \in \N}~ \Phi_{B,P}^n(\boldsymbol{0})~.
	\]
\end{restatable}%

\subsection{Healthiness Conditions of FPS Transformers}\label{ssec:sem_props}
In this section we show basic, yet important, properties which follow from~\cite{Kozen79}.
For instance, for any input FPS $F$, the semantics of a program cannot yield as output an FPS with a mass larger than $\mass{F}$, i.e.\ \emph{programs cannot create mass}.%
\begin{restatable}[Mass Conservation]{theorem}{masscons}\label{thm:mass}
For every $P \in \pgcl$ and $F \in \ps$, we have $\mass{\sem{P}(F)\vphantom{\big(}} \lleq \mass{F}$.
\end{restatable}%
\noindent{}%
A program $P$ is called \emph{mass conserving} if $\mass{\sem{P}(F)} = \mass{F}$ for all $F \in \ps$. Mass conservation has important implications for FPS transformers acting on PGFs: given as input a PGF, the semantics of a program yields a PGF.%
\begin{corollary}[PGF Transformers]
	For every $P \in \pgcl$ and $G \in \pgf$, we have $\semapp{P}{G} \in \pgf$.
\end{corollary}%
\noindent{}%
Restricted to $\pgf$, our semantics hence acts as a subdistribution transformer. 
Output masses may be smaller than input masses. 
The probability of non-termination of the programs is captured by the \enquote{missing} probability mass.

As observed in~\cite{Kozen79}, semantics of probabilistic programs are fully defined by their effects on point masses, thus rendering probabilistic program semantics linear.
In our setting, this generalizes to linearity of our FPS transformers.%
\begin{definition}[Linearity]
	Let $F, G \in \ps$ and $r \in \PosRealsInf$ be a scalar. 
	The function $\psi \colon \ps \to \ps$ is called a \emph{linear transformer} (or simply \emph{linear}), if
	\[
		\psi(r \cdot F + G) \eeq r \cdot \psi(F) \pplus \psi(G)~.
	\]%
\end{definition}%
\begin{restatable}[Linearity of \pgcl Semantics]{theorem}{linearity}\label{lem:lin_funcs} 
	For every program $P$ and guard~$B$, the functions $\restrict{\:\cdot\:}{B}$ and $\sem{P}$ are linear. 
	Moreover, the unfolding operator $\Phi_{B,P}$ maps linear transformers onto linear transformers.
\end{restatable}%
\noindent{}As a final remark, we can unroll \while loops:
\begin{restatable}[Loop Unrolling]{lemma}{unrolling}\label{lem:loop_unrolling}
	For any 
	FPS $F$,
	\[
		\semapp{\WHILEDO{B}{P}}{F} \eeq \restrict{F}{\neg B} \pplus \semapp{\WHILEDO{B}{P}}{\sem{P}\big(\restrict{F}{B}\big)}~.
	\]%
\end{restatable}

\subsection{Embedding into Kozen's Semantics Framework}\label{ssec:kozen_relation}
Kozen~\cite{Kozen79} defines a generic way of giving distribution transformer semantics based on an abstract measurable space $(X^n, M^{(n)})$.
Our FPS semantics instantiates his generic semantics.
The state space we consider is $\N^k$, so that $(\N^k, \mathcal{P}(\N^k))$ is our measurable space.%
\footnote{We note that we want each point $\sigma$ to be measurable, which enforces a \emph{discrete} measurable space.}
A measure on that space is a countably-additive function $\mu \colon \mathcal{P}(\N^k) \to [0, \infty]$ with $\mu(\emptyset) = 0$.
We denote the set of all measures on our space by~$\mathcal{M}$.
Although, we represent measures by FPSs, the two notions are in bijective correspondence $\tau \colon \ps \to \mathcal{M}$, given by
\[
  \tau(F) \eeq \lambda S \mathpunct{.} \sum_{\sigma  \in S} [\sigma]_F ~.
\]
This map preserves the linear structure and the order $\preceq$.

Kozen's syntax~\cite{Kozen79} is slightly different from \pgcl. 
We compensate for this by a translation function $\mathfrak{T}$, which maps \pgcl programs to Kozen's.
The following theorem shows that our semantics agrees with Kozen's semantics.\footnote{Note that Kozen regards a program $P$ itself as a function $P \colon \mathcal{M} \to \mathcal{M}$.}%
\begin{restatable}{theorem}{kozenrel}\label{thm:rel_to_kozen} The \ps semantics of \pgcl 
is an instance of Kozen's semantics, i.e.~for all \pgcl programs $P$, we have%
\[\tau \circ \sem{P} = \mathfrak{T}(P) \circ \tau~.\]%
\end{restatable}%
\noindent{}%
Equivalently, the following diagram commutes:
	\begin{center}
		\begin{tikzcd}
		\ps \arrow["\sem{P}", swap, dd] \arrow[rr, "\tau"] &&  \mathcal{M} \ar[dd, "\mathfrak{T}(P)"]\\
		&&\\
		\ps \arrow["\tau", swap, rr] && \mathcal{M}
		\end{tikzcd}
	\end{center}
For more details about the connection between FPSs and measures, as well as more information about the actual translation, see \cref{app:kozen}.

\section{Analysis of Probabilistic Programs}\label{sec:overapprox}
Our PGF semantics enables the representation of the effect of a \pgcl program on a given PGF. 
As a next step, we investigate to what extent a program analysis can exploit such PGF representations. 
To that end, we consider the overapproximation with loop invariants (\cref{ssec:overapprox}) and provide examples showing that checking whether an FPS transformer overapproximates a loop can be checked with computer algebra tools. 
In addition, we determine a subclass of \pgcl programs whose effect on an arbitrary input state is ensured to be a rational PGF encoding a phase-type distribution (\cref{ssec:rational_pgf}).

\subsection{Invariant-style Overapproximation of Loops}\label{ssec:overapprox}
	In this section, we seek to overapproximate loop semantics, i.e.\ for a given loop $W = \WHILEDO{B}{P}$, we want to find a (preferably simple) FPS transformer $\psi$, such that $\sem{W} \sqsubseteq \psi$, meaning that for any input~$G$, we have $\semapp{W}{G} \preceq \psi(G)$ (cf.~\Cref{def:semantics}).
	Notably, even if $G$ is a PGF, we do not require $\psi(G)$ to be one. 
	Instead, $\psi(G)$ can have a mass larger than one. This is fine, because it still overapproximates the actual semantics coefficient-wise.
	Such overapproximations immediately carry over to reading off expected values (cf.~\Cref{sec:gf}), for instance%
 	\[
 		\tfrac{\partial}{\partial X}\semapp{W}{G}(\vec{1})
 		\quad\lleq\quad
 		\tfrac{\partial}{\partial X}\psi(G)(\vec{1})~.
 	\]%
	We use invariant-style reasoning for verifying that a \emph{given} $\psi$ overapproximates the semantics of~$\sem{W}$.
	For that, we introduce the notion of a \emph{superinvariant} and employ Park's Lemma, a well-known concept from fixed point theory, to obtain a conceptually simple proof rule for verifying overapproximations of while loops.%
	\begin{restatable}[Superinvariants and Loop Overapproximations]{theorem}{overapprox}\label{thm:overapprox}
		Let $\Phi_{B,P}$ be the unfolding operator of $\WHILEDO{B}{P}$ (cf.~\textnormal{Def.~\ref{def:semantics}}) and $\psi\colon \ps \to \ps$. Then
		\[
			\Phi_{B,P}(\psi) \ssqsubseteq \psi \qimplies \sem{\WHILEDO{B}{P}} \ssqsubseteq \psi~.
		\]
	\end{restatable}%
\noindent{}%
We call a $\psi$ satisfying $\Phi_{B,P}(\psi) \sqsubseteq \psi$ a \emph{superinvariant}.
	We are interested in linear superinvariants, as our semantics is also linear (cf.~\Cref{lem:lin_funcs}).
	Furthermore, linearity allows to define $\psi$ solely in terms of its effect on monomials, which makes reasoning considerably simpler:

	\begin{restatable}{corollary}{superinvs}\label{cor:linear_superinvs}
		Given a function $f\colon \mon{\bvec{X}} \to \ps$, let the \emph{linear extension $\hat{f}$ of $f$} be defined by
		\[
			\hat{f} \colon\quad \ps \tto \ps, \quad F \mmapsto \sum_{\sigma  \in \N^k} \extract{\sigma}{F} f(\bvec{X}^\sigma)~.
		\]%
		Let $\Phi_{B,P}$ be the unfolding operator of $\WHILEDO{B}{P}$. Then
		\[
			\forall\, \sigma \in \N^k\colon~~ \Phi_{B,P}(\hat{f})(\bvec{X}^\sigma) \ssqsubseteq \hat{f}(\bvec{X}^\sigma) \qqimplies \sem{\WHILEDO{B}{P}} \ssqsubseteq \hat{f}~.
		\]
	\end{restatable}%

\noindent{}%
We call an $f$ satisfying the premise of the above corollary a \emph{superinvariantlet}.
	Notice that superinvariantlets and their extensions agree on monomials, i.e.\ $f(\bvec{X}^\sigma) = \hat{f}(\bvec{X}^\sigma)$.
	Let us examine a few examples for superinvariantlet-reasoning.%
\label{sec:examples_over_approx}%
\begin{example}[Verifying Precise Semantics]\label{ex:approx_geo}
	In %
	\Pref{prog:geometric}, in each iteration, a fair coin flip determines the value of $\codify{x}$. Subsequently, $\codify{c}$ is incremented by $1$.
	Consider the following superinvariantlet:
	\[
	f(X^{i}C^{j}) \eeq C^{j} \cdot 
		\begin{cases}
			\frac{C}{2-C}, & \text{if}~i = 1;\\
			X^{i},& \text{if}~ i \neq 1.
		\end{cases}
	\]
	To verify that $f$ is indeed a superinvariantlet, we have to show that
	\begin{align*}
		\Phi_{B,P}(\hat{f})(X^iC^j) & \eeq \constrain{X^iC^j}{x \neq 1} + \hat{f}\left(\sem{P}\big(\constrain{X^iC^j}{x=1}\big)\right) \\
		& \overset{!}{\ssqsubseteq} \hat{f} \left( X^i C^j \right)~.
	\end{align*}
	For $i\neq 1$, we get
	\begin{align*}
		\Phi_{B,P}(\hat{f})(X^i C^j) 
		&\eeq \constrain{X^i C^j}{x \neq 1} + \hat{f}(\semapp{P}{0}) \\
		&\eeq X^i C^j \eeq f(X^i C^j) \eeq \hat{f}(X^i C^j)~.
	\end{align*}
	For $i=1$, we get
	\begin{align*}
	\Phi_{B,P}(\hat{f})(X^1C^j) 
	&\eeq  \hat{f} \left(\tfrac{1}{2} X^0 C^{j+1} + \tfrac{1}{2} X^1 C^{j+1} \right) \\
	&\eeq \tfrac{1}{2} f\left( X^0 C^{j+1}\right) + \tfrac{1}{2} f\left( X^1 C^{j+1} \right)
	\tag{by linearity of $\hat{f}$} \\
	&\eeq \tfrac{C^{j+1}}{2-C}
	\eeq f\left(X^1C^j\right) \eeq \hat{f}\left( X^1C^j\right)~.
	\tag{by definition of $f$}
	\end{align*}
	Hence, \Cref{cor:linear_superinvs} yields
	$\sem{W}(X) \ssqsubseteq f \left( X \right) \eeq \tfrac{C}{2-C}$.

	For this example, we can state even more.
	As the program is almost surely terminating, and $\mass{f(X^i C^j)} = 1$ for all $(i,j) \in \N^2$, we conclude that $\hat{f}$ is exactly the semantics of $W$, i.e. $\hat{f} = \sem{W}$.
	\hfill$\triangle$
\end{example}%
\program{geometric}{Geometric distribution generator.}{t}%
\begin{example}[Verifying Proper Overapproximations]\label{ex:rw}
	\Pref{prog:random_walk} models a one dimensional, left-bounded random walk.
	Given an input $(i,j) \in \N^2$, it is evident that this program can only terminate in an even (if $i$ is even) or odd (if $i$ is odd) number of steps. 
	This information can be encoded into the following superinvariantlet:
	\begin{align*}
		f(X^0C^j) &\eeq C^j \quad \text{and}\\
		 f(X^{i+1}C^j) &\eeq C^j \cdot  
		\begin{cases}
			\frac{C}{1-C^2},& \text{if}~ i ~\text{is odd;}\\
			\frac{1}{1-C^2},& \text{if}~ i ~\text{is even.}
		\end{cases}
	\end{align*}
	It is straightforward to verify that $f$ is a \emph{proper} superinvariantlet (proper because $\frac{C}{1-C^2} = C + C^3 + C^5 + \ldots$ is \emph{not} a PGF) and hence $f$ \emph{properly} overapproximates the loop semantics.
	\program{random_walk}{Left-bounded 1-dimensional random walk.}{t}
	Another superinvariantlet for \Pref{prog:random_walk} is given by
	\[
		h(X^{i}C^{j}) \eeq C^{j} \cdot 
		\begin{cases}
			\left(\frac{1 - \sqrt{1 - C^2} }{C}\right)^{i},& \text{if}~ i \ge 1;\\
			1,& \text{if}~ i = 0.
		\end{cases}
	\]
	Given that the program terminates almost-surely~\cite{Hurd02} and that $h$ is a superinvariantlet yielding only PGFs, it follows that the extension of $h$ is \emph{exactly} the semantics of \Pref{prog:random_walk}. An alternative derivation of this formula for the case $h(X)$ can be found, e.g., in~\cite{Icard}.
	
	For both $f$ and~$h$, we were able to prove that they are indeed superinvariantlets \emph{automatically}, using the computer algebra library \textsf{SymPy}~\cite{SymPy}. The code is included in Appendix \ref{app:over_approx} (\Pref{prog:sympy}).\hfill$\triangle$
\end{example}%
\begin{example}[Proving Non-almost-sure Termination]\label{ex:non_ast}
	In \Pref{prog:non-termination}, 
	the branching probability of the choice statement depends on the value of a program variable. This notation is just syntactic sugar, as this behavior can be mimicked by loop constructs together with coin flips~\cite[pp. 115f]{Batz18}.
	
	To prove that \Pref{prog:non-termination} does \emph{not} terminate almost-surely, we consider the following superinvariantlet:
	\[
	 	f(X^i) \eeq 1-\frac{1}{e}\cdot\sum_{n=0}^{i-2}\frac{1}{n!}~, \qquad \text{where } e = 2.71828 \ldots \text{ is Euler's number.}
	\]
	Again, the superinvariantlet property was \emph{verified automatically}, here using \textsf{Mathematica}~\cite{Mathematica}.
	Now, consider for instance \mbox{$f(X^3) = 1 - \frac{1}{e}\cdot \left(\frac{1}{0!} + \frac{1}{1!}\right) = 1 - \frac{2}{e} < 1$}.
	This proves, that the program terminates on $X^3$ with a probability strictly smaller than 1, witnessing that the program is not almost surely terminating. \hfill$\triangle$%
	\program{non-termination}{A non-almost-surely terminating loop.}{t}
\end{example}%

\subsection{Rational PGFs}\label{ssec:rational_pgf}
In several of the examples from the previous sections, we considered PGFs which were \emph{rational functions}, that is, fractions of two polynomials. Since those are a particularly simple class of PGFs, it is natural to ask which programs have rational semantics.
In this section, we present a semantic characterization of a class of \while-loops whose output distribution is a (multivariate) \emph{discrete phase-type} dsitribution~\cite{phasetype_phd,Neuts81}. This implies that the resulting PGF of such programs is an effectively computable rational function for any given input state. Let us illustrate this by an example.

\begin{example}[Dueling Cowboys]
	\label{ex:dueling_cowboys} \Pref{prog:dueling_cowboys} models two dueling cowboys~\cite{DBLP:series/mcs/McIverM05}.
	\program{dueling_cowboys}{Dueling cowboys.}{t}
	The hit chance of the first cowboy is $a$ percent and the hit chance of the second cowboy is $b$ percent, where $a,b \in [0,1]$.\footnote{These are \emph{not} program variables.}
	The cowboys shoot at each other in turns, as indicated by the variable $\codify{t}$, until one of them gets hit ($\codify{x}$ is set to $\codify{1}$). The variable $\codify{c}$ counts the number of shots of the first cowboy and $\codify{d}$ those of the second cowboy.
	
	We observe that \Pref{prog:dueling_cowboys} is somewhat independent of the value of $\codify{c}$, in the sense that moving the statement $\codify{c := c + 1}$ to either immediately before or after the loop, yields an equivalent program. In our notation, this is expressed as $\sem{W}(C \cdot H) = C \cdot \sem{W}(H)$ for all PGFs $H$. By symmetry, the same applies to variable $\codify{d}$. Unfolding the loop once on input $1$, yields
	\[
		\sem{W}(1) \eeq (1-a)C \cdot \sem{W}(T) + aCX~.
	\]
	A similar equation for $\sem{W}(T)$ involving $\sem{W}(1)$ on its right-hand side holds.
	This way we obtain a system of two linear equations, although the program itself is infinite-state.
	The linear equation system has a unique solution $\sem{W}(1)$ in the field of rational functions over the variables $C,D,T$, and $X$ which is the PGF
	\[
		G \ccoloneqq \frac{aCX + (1-a)bCDTX}{1-(1-b)(1-a)CD}~.
	\]
	From $G$ we can easily read off the following: The probability that the first cowboy wins ($\codify{x = 1}$ and $\codify{t = 0}$) equals $\frac{a}{1 - (1-a)(1-b)}$, and the expected total number of shots of the first cowboy is
	$
	\tfrac{\partial}{\partial C} G(1) = \frac{1}{a + b - a b}.
	$
	Notice that this quantity equals $\infty$ if $a$ and $b$ are both zero, i.e.\ if both cowboys have zero hit chance.
	
	If we write $G_\bvec{V}$ for the PGF obtained by substituting all but the variables in $\bvec{V}$ with $1$,
	then we moreover see that $G_C \cdot G_D \neq G_{C,D}$.
	This means that $C$ and $D$ (as random variables) are stochastically dependent.
	\hfill$\triangle$
\end{example}%
\noindent{}%
The distribution encoded in the PGF $\sem{W}(1)$ is a discrete phase-type distribution. Such distributions are defined as follows: A \textit{Markov reward chain} is a Markov chain where each state is augmented with a reward vector in $\N^k$. By definition, a (discrete) distribution on $\mathbb{N}^k$ is of phase-type iff it is the distribution of the total accumulated reward vector until absorption in a Markov reward chain with a single absorbing state and a finite number of transient states. In fact, \Pref{prog:dueling_cowboys} can be described as a Markov reward chain with two states ($X^0 T^0$ and $X^0 T^1$) and 2-dimensional reward vectors corresponding to the ``counters'' $(\codify{c}, \codify{d})$: the reward in state $X^0 T^0$ is $(1,0)$ and $(0,1)$ in the other state.

Each \pgcl program describes a Markov reward chain~\cite{DBLP:journals/pe/GretzKM14}. It is not clear which (non-trivial) syntactical restrictions to impose to guarantee for such chains to be finite. 
In the remainder of this section, we give a characterization of \while-loops that are equivalent to finite Markov reward chains. The idea of our criterion is that each variable has to fall into one of the following two categories:%
\begin{definition}[Homogeneous and Bounded Variables]
	\label{def:bounded_variables}
	Let $P \in \pgcl{}$ be a program, $B$ be a guard and $\codify{x}_i$ be a program variable. Then:%
	\begin{itemize}
		\item $\codify{x}_i$ is called \emph{homogeneous} for $P$ if $\sem{P}(X_i \cdot G) = X_i \cdot \sem{P}(G)$ for all $G \in \pgf$.
		\item $\codify{x}_i$ is called \emph{bounded} by $B$ if the set $\{ \sigma_i \mid \sigma \in B \}$ is finite.
	\end{itemize}%
\end{definition}%
\noindent{}%
Intuitively, homogeneity of $\codify{x}_i$ means that it does not matter whether one increments the variable before or after the execution of $P$. Thus, a homogeneous variable \emph{behaves like an increment-only counter} even if this may not be explicit in the syntax. In Example \ref{ex:dueling_cowboys}, the variables $\codify{c}$ and $\codify{d}$ in \Pref{prog:dueling_cowboys} are homogeneous (for both the loop-body and the loop itself). Moreover, $\codify{x}$ and $\codify{t}$ are clearly bounded by the loop guard. We can now state our characterization.

\begin{definition}[HB Loops]
	\label{def:hom_bounded_loop}
	A loop $\WHILEDO{B}{P}$ is called \emph{homogeneous-bounded (HB)} if for all program states $\sigma \in B$, the PGF $\sem{P}(\bvec{X}^\sigma)$ is a polynomial and for all program variables $\codify{x}$ it \emph{either} holds that
		\begin{itemize}
			\item $\codify{x}$ is homogeneous for $P$ and the guard $B$ is independent of $\codify{x}$, or that
			\item $\codify{x}$ is bounded by the guard $B$.
		\end{itemize}
\end{definition}%
\noindent{}%
In an HB loop, all the possible valuations of the bounded variables satisfying~$B$ span the \emph{finite} transient state space of a Markov reward chain in which the dimension of the reward vectors equals the number of homogeneous variables.
The additional condition that $\sem{P}(\bvec{X}^\sigma)$ is a polynomial ensures that there is only a finite amount of terminal (absorbing) states.
Thus, we have the following:%
\begin{restatable}{proposition}{phasetype}
	Let $W$ be a while-loop. Then $\sem{W}(\bvec{X}^\sigma)$ is the (rational) PGF of a multivariate discrete phase-type distribution if and only if $W$ is equivalent to an HB loop that almost-surely terminates on input $\sigma$.
\end{restatable}%
\noindent{}%
To conclude, we remark that there are various simple \emph{syntactic} conditions for HB loops: For example, if $P$ is loop-free, then $\sem{P}(\bvec{X}^\sigma)$ is always a polynomial. Similarly, if $\codify{x}$ only appears in assignments of the form $\codify{x := x + } k$, $k \geq 0$, then $\codify{x}$ is homogeneous. Such updates of variables are e.g. essential in \emph{constant probability programs}~\cite{DBLP:conf/cade/GieslGH19}. The crucial point is that such conditions are only sufficient but not necessary. Our \emph{semantic} conditions thus capture the essence of phase-type distribution semantics more adequately while still being reasonably simple (albeit --- being non-trivial semantic properties --- undecidable in general).


\section{Conclusion}\label{sec:conclusion}
%
%
We have presented a denotational distribution transformer semantics for probabilistic while-programs where the denotations are generation functions (GFs).
Moreover, we have provided a simple invariant-style technique to prove that a given GF over{\-}approximates the program's semantics and identified a class of (possibly infinite-state) programs whose semantics is a rational GF ecoding a phase-type distribution. 
Directions for future work include the (semi-)automated synthesis of invariants and the development of notions on how precise  overapproximations by invariants actually are.

%
%
%
\bibliographystyle{splncs04}
\bibliography{references}


\appendix
\section{Proofs of \Cref{sec:semantics}}
\subsection{Proofs of \Cref{ssec:semantics}}

\cocpos*
\begin{proof}
	We start by showing that $(\ps, \preceq)$ is a partial order.
	Let $F, G, H \in \ps$, $\sigma \in \N^k$.
	For reflexivity, consider the following:
	\begin{align*}
	& G \ppreceq G \\
	\text{iff} \qquad & \forall \sigma \in \N^k \colon \extract{\sigma}{G} \lleq \extract{\sigma}{G} \\
	\text{iff} \qquad & \textnormal{true}~.
	\end{align*}
	For antisymmetry, consider the following:
	\begin{align*}
	& G \ppreceq H ~ \textnormal{and} ~ H \ppreceq G \\
	\text{implies} \qquad & \forall \sigma \in \N^k \colon \extract{\sigma}{G} \lleq \extract{\sigma}{H} \quad \textnormal{and} \quad \extract{\sigma}{H} \lleq \extract{\sigma}{G} \\
	\text{implies} \qquad & \forall \sigma \in \N^k \colon \extract{\sigma}{G} \eeq \extract{\sigma}{H} \\
	\text{implies} \qquad & \quad G \eeq H~.
\end{align*}
For transitivity, consider the following:
\begin{align*}
      & G \ppreceq H \quad \textnormal{and} \quad H \ppreceq F \\
	\text{implies} \qquad & \forall \sigma \in \N^k \colon \extract{\sigma}{G} \lleq \extract{\sigma}{H} \quad \textnormal{and} \quad \extract{\sigma}{H} \lleq \extract{\sigma}{F} \\
	\text{implies} \qquad & \forall \sigma \in \N^k \colon \extract{\sigma}{G} \lleq \extract{\sigma}{F} \\
	\text{implies} \qquad & \quad G \ppreceq F~.
	\end{align*}
	Next, we show that every set $S \subseteq \ps$ has a supremum
	\[
		\sup S \eeq \sum_{\sigma \in \N^k} \sup_{F \in S}  \extract{\sigma}{F} \bvec{X}^\sigma~
	\]
	 in $\ps$. In particular, notice that $\sup \emptyset = \sum_{\sigma \in \N^k} 0 \cdot \bvec{X}^{\sigma}$.
	 The fact that $\sup S \in \ps$ is trivial since $\sup_{F \in S}  \extract{\sigma}{F} \in \PosRealsInf$ 
	 for every $\sigma \in \N^k$. Furthermore, the fact that $\sup S$ is an upper bound on $S$ is immediate
	 since $\preceq$ is defined coefficient-wise. Finally, $\sup S$ is also the \emph{least} upper bound, since, by 
	 definition of $\preceq$, we have $\extract{\sigma}{\sup S} = \sup_{F \in S} \extract{\sigma}{F}$. 
	
\end{proof}

The following proofs rely on the Monotone Sequence Theorem (MST), which we recall here:
If $(a_n)_{n\in\N}$ is a monotonically
increasing sequence in $\PosRealsInf$, then $\sup_n a_n = \lim_{n \rightarrow \infty} a_n$. In particular, if $(a_n)_{n\in\N}$ and $(b_n)_{n\in\N}$ are monotonically increasing sequences in $\PosRealsInf$, then
\[
  \sup_n a_n + \sup_n b_n 
  \eeq 
  \lim_{n \rightarrow \infty} a_n 
  + \lim_{n \rightarrow \infty} b_n
  \eeq \lim_{n \rightarrow \infty} a_n + b_n
  \eeq \sup_n a_n + b_n~.
\]

\cont*
\begin{proof}
	By induction on the structure of $P$. Let $S= \left\{F_1, F_2,\ldots \right\}$ be an increasing $\omega$-chain in $\ps$.
	First, we consider the base cases. \\ \\
\noindent
\emph{The case $P = \pskip{}$.} We have 
\[
	\semapp{P}{\sup S} \eeq \sup S \eeq \sup_{F \in S} ~ \{ F \} \eeq \sup_{F \in S} ~ \{ \semapp{P}{F} \}~.
\]
\emph{The case $P = x_i \coloneqq E$.}	Let $\sup S = \hat{G} = \sum_{\sigma \in \N^k} \extract{\sigma}{\hat{G}} \cdot X^\sigma$, where
for each $\sigma \in \N^k$ we have
$\extract{\sigma}{\hat{G}} = \sup_{F \in S} \extract{\sigma}{F}$. We calculate

\begin{align*}
   &\semapp{P}{\sup S} \\
   \eeq & \semapp{P}{\hat{G}} \\
   \eeq & \semapp{P}{ \sum_{\sigma \in \N^k} \extract{\sigma}{\hat{G}} \cdot X^\sigma} \\
   \eeq & \semapp{P}{ \sum_{\sigma \in \N^k} \extract{\sigma}{\hat{G}} \cdot X_1^{\sigma_1} \cdots  X_i^{\sigma_i} \cdots X_k^{\sigma_k}} \\
   \eeq &  \semapp{P}{ \sum_{\sigma \in \N^k} \extract{\sigma}{\hat{G}} \cdot X_1^{\sigma_1} \cdots  X_i^{\sigma_i} \cdots X_k^{\sigma_k}} \\
    \eeq & \sum_{\sigma \in \N^k} \extract{\sigma}{\hat{G}} \cdot X_1^{\sigma_1} \cdots  X_i^{\evalstate{\sigma}{E}} \cdots X_k^{\sigma_k} \\
    \eeq & \sum_{\sigma \in \N^k} \left(\sup_{F \in S} \extract{\sigma}{F} \right) \cdot X_1^{\sigma_1} \cdots  X_i^{\evalstate{\sigma}{E}} \cdots X_k^{\sigma_k} 
        \\
    \eeq & \sup_{F \in S} \sum_{\sigma \in \N^k}  \extract{\sigma}{F}  \cdot X_1^{\sigma_1} \cdots  X_i^{\evalstate{\sigma}{E}} \cdots X_k^{\sigma_k}
    \tag{$\sup$ on $\ps$ is defined coefficient--wise} \\
    \eeq & \sup_{F \in S}   \semapp{P}{ \sum_{\sigma \in \N^k} \extract{\sigma}{F} \cdot X_1^{\sigma_1} \cdots  X_i^{\sigma_i} \cdots X_k^{\sigma_k}} \\
    \eeq & \sup_{F \in S}   \semapp{P}{ F} 
\end{align*}
As the induction hypothesis now assume that for some arbitrary, but fixed, 
programs $P_1$, $P_2$ and all increasing $\omega$-chains $S_1, S_2$ in $\ps$ it holds that both
\[
   \semapp{P_1}{\sup S_1} \eeq \sup_{F\in S_1} \semapp{P_1}{F}
   \qquad \textnormal{and} \qquad
      \semapp{P_2}{\sup S_2} \eeq \sup_{F\in S_2} \semapp{P_2}{F}~.
\]

We continue with the induction step. \\ \\
\noindent
\emph{The case $P = \PCHOICE{P_1}{p}{P_2}$}. We have
\begin{align*}
   &\semapp{P}{\sup S}  \\
   \eeq & p \cdot \semapp{P_1}{\sup S} + (1-p) \cdot \semapp{P_2}{\sup S} \\
   \eeq & p \cdot \left( \sup_{F \in S} \semapp{P_1}{F} \right) + 
                 (1-p) \cdot \left( \sup_{F \in S} \semapp{P_2 }{F} \right) 
    \tag{I.H.\ on $P_1$ and $P_2$}\\
   \eeq &  \left( \sup_{F \in S} p \cdot  \semapp{P_1}{F} \right) + 
    \left( \sup_{F \in S} (1-p) \cdot  \semapp{P_2 }{F} \right) 
    \tag{scalar multiplication is defined point--wise}
    \\  
    \eeq&\sup_{F \in S}  \left( p \cdot  \semapp{P_1}{F} + 
    (1-p) \cdot \semapp{P_2 }{F} \right) 
    \tag{apply MST coefficient--wise.} \\
    \eeq&\sup_{F \in S} \semapp{\PCHOICE{P_1}{p}{P_2}}{F} \\
    \eeq&\sup_{F \in S} \semapp{P}{F}~.
\end{align*}
\emph{The case $P = \ITE{B}{P_1}{P_2}$.} We have
\begin{align*}
   \semapp{P}{\sup S} \\
   \eeq & \semapp{P_1}{\restrict{\sup S}{B}} + \semapp{P_2}{\restrict{\sup S}{\neg B}} \\
   \eeq & \semapp{P_1}{\sup_{F \in S}  \big( \restrict{F}{B} \big)} + \semapp{P_2}{\sup_{F \in S} \big( \restrict{F}{\neg B} \big)}
   \tag{restriction defined coefficient--wise} \\
   \eeq &\sup_{F \in S}  \semapp{P_1}{ \restrict{F}{B}} + \sup_{F \in S} \semapp{P_2}{\restrict{F}{\neg B}}
   \tag{I.H.\ on $P_1$ and $P_2$} \\
   \eeq &\sup_{F \in S} \left( \semapp{P_1}{ \restrict{F}{B}} + \semapp{P_2}{\restrict{F}{\neg B}} \right)
   \tag{apply MST coefficient--wise} \\
   \eeq& \sup_{F \in S} \semapp{\ITE{B}{P_1}{P_2}}{F} \\
   \eeq& \sup_{F \in S} \semapp{P}{F}~.
\end{align*}
\emph{The case $P = \while{}(B)\{P_1\}$.}
Recall that for every $G \in \ps$,
\begin{align*}
   \semapp{P}{G} &\eeq\left(\lfp \, \phifct{B}{P_1} \right) (G)  \\
   & \eeq \big( \sup_{n \in N} \phifctpowerapp{B}{P_1}{n}{\mathbf{0}}  \big) (G)~.
\end{align*}
Hence, it suffices to show that 
\[
 \left( \sup_{n \in N} \phifctpowerapp{B}{P_1}{n}{\mathbf{0}}  \right) (\sup S)
 \eeq
 \sup_{F\in S} \left( \big( \sup_{n \in N} \phifctpowerapp{B}{P_1}{n}{\mathbf{0}}  \big) (F)\right) ~.
\]
Assume for the moment that for every $n \in \N$ and all increasing $\omega$-chains $S$ in $\ps$,
\begin{align}\label{eq:phi_n_cont}
  \left(  \phifctpowerapp{B}{P_1}{n}{\mathbf{0}} \right)\left( \sup S \right) 
  \eeq  \sup_{F \in S} \left(  \phifctpowerapp{B}{P_1}{n}{\mathbf{0}} \right) (F) ~.
\end{align}
We then have
\begin{align*}
   & \left( \sup_{n \in \N} \phifctpowerapp{B}{P_1}{n}{\mathbf{0}}  \right) (\sup S)  \\
   \eeq & \sup_{n \in \N} \left( \phifctpowerapp{B}{P_1}{n}{\mathbf{0}}(\sup S)  \right)
   \tag{$\sup$ for $\phifct{B}{P_1}$ is defined point--wise} \\
   \eeq& \sup_{n \in \N} \sup_{F \in S} \left( \phifctpowerapp{B}{P_1}{n}{\mathbf{0}}(F)  \right)
   \tag{Equation \ref{eq:phi_n_cont}} \\
   \eeq &\sup_{F \in S} \sup_{n \in \N}  \left( \phifctpowerapp{B}{P_1}{n}{\mathbf{0}}(F)  \right)
   \tag{swap suprema} \\
    \eeq &\sup_{F \in S} \left(\big( \sup_{n \in \N}  \left( \phifctpowerapp{B}{P_1}{n}{\mathbf{0}} \big) (F)\right)  \right)
   \tag{$\sup$ for $\phifct{B}{P_1}$ is defined point--wise}~,  \\
\end{align*} 
which is what we have to show. It remains to prove Equation~\ref{eq:phi_n_cont} \mbox{by induction on $n$.} \\ \\
\noindent
\emph{Base case $n=0$.} We have
\[
   \left(  \phifctpowerapp{B}{P_1}{0}{\mathbf{0}} \right)\left( \sup S \right)  
   \eeq 
   \sup S
   \eeq 
   \sup_{F \in S} F 
   \eeq 
   \sup_{F \in S}
   \left(  \phifctpowerapp{B}{P_1}{0}{\mathbf{0}} \right)\left( F \right)~.  
\]

\noindent
\emph{Induction step.} We have
\begin{align*}
   &\left(  \phifctpowerapp{B}{P_1}{n+1}{\mathbf{0}} \right)\left( \sup S \right)  \\
   \eeq& \phifctapp{B}{P_1}{\phifctpowerapp{B}{P_1}{n}{\mathbf{0}}}\left( \sup S \right)  \\
   \eeq&
   \restrict{\sup S}{\neg B} 
   + \phifctpowerapp{B}{P_1}{n}{\mathbf{0}} \left( \semapp{P_1}{\restrict{\sup S}{B}} \right)
   \tag{Def.\ of $\phifct{B}{P_1}$} \\
   \eeq & \restrict{\sup S}{\neg B} 
   + \phifctpowerapp{B}{P_1}{n}{\mathbf{0}} \left( \sup_{F \in S} \semapp{P_1}{\restrict{F}{B}} \right)
   \tag{I.H.\ on $P_1$} \\
   \eeq & \restrict{\sup S}{\neg B} 
   + \sup_{F \in S}\phifctpowerapp{B}{P_1}{n}{\mathbf{0}} \left(  \semapp{P_1}{\restrict{F}{B}} \right)
   \tag{I.H.\ on $n$} \\
   	\eeq & \sup_{F \in S} \left( \restrict{F}{\neg B} 
   	+  \phifctpowerapp{B}{P_1}{n}{\mathbf{0}} \left(\semapp{P_1}{\restrict{F}{B}} \right) \right)
   	\tag{apply MST} \\
   	\eeq &\sup_{F \in S}  \left(  \phifctpowerapp{B}{P_1}{n+1}{\mathbf{0}} \right)\left( F \right)
   	\tag{Def.\ of $\phifct{B}{P_1}$}~. \\
\end{align*}
This completes the proof.\qedhere
\end{proof}

\welldef*
\begin{proof}
	First, we show that the unfolding operator $\Phi_{B,P}$ is $\omega$-continuous. For that, let $f_1 \sqsubseteq f_2 \sqsubseteq \ldots$ be an $\omega$-chain in $\ps \to \ps$. Then,
	\begin{align*}
		\Phi_{B,P}\left(\sup_{n\in \N}\{f_n\}\right) &\eeq \lambda G. \left<G\right>_{\neg B} + \left(\sup_{n\in \N}\{f_n\}\right) (\sem{P}(\left<G\right>_B))\\
		&\eeq\lambda G. \left<G\right>_{\neg B} + \sup_{n\in \N}\{f_n(\sem{P}(\left<G\right>_B))\}
		\tag{$\sup$ on $\ps \to \ps$ is defined point--wise}\\
		&\eeq\sup_{n\in \N}\left\{\lambda G. \left<G\right>_{\neg B} + f_n(\sem{P}(\left<G\right>_B))\right\}
		\tag{apply monotone sequence theorem coefficient-wise}\\
		&\eeq \sup_{n\in \N}\left\{\Phi_{B,P}(f_n)\right\}~.
		\tag{Def.\ of $\phifct{B}{P}$}
	\end{align*}
	Since $\Phi_{B,P}$ is $\omega$-continuous and $(\ps \to \ps, \sqsubseteq)$ forms a complete lattice (\Cref{lem:completeness_of_pos}), we get by the Kleene fixed point Theorem~\cite{DBLP:journals/ipl/LassezNS82} that $\Phi_{B,P}$ has  a unique least fixed point given by $\sup_{n \in \N} \Phi_{B,P}^n(\boldsymbol{0})$.
\end{proof}
\masscons*
\begin{proof}
By induction on the structure of $P$. For the loop--free cases, this is straightforward.
For the case $P = \while{}(B)\{P_1\}$, we proceed as follows. For every $r \in \PosRealsInf$, we define
the set 
\[
   \ps_r \eeq \left\{ F \in \ps ~|~ \mass{F} \leq r \right\}
\]
of all FPSs whose mass is at most $r$. 
First, we define the restricted unfolding operator 
\[
    \phifctres{B}{P_1}{r} \colon (\ps_r \to \ps_r) \to (\ps_r \to \ps_r), \quad \psi  \mapsto \phifctapp{B}{P_1}{\psi}~.
\]
Our induction hypothesis on $P_1$ implies that $\phifctres{B}{P_1}{r}$ is well--defined.

It is now only left to show that $\left( \ps_r, \, \preceq \right)$ is an $\omega$-complete partial order, because then $\phifctres{B}{P_1}{r}$ has a least fixed point in $\ps_r$ for every $r \in \PosRealsInf$. 
The theorem then follows by letting $r = \mass{G}$, because
\[
    \left( \lfp \, \phifct{B}{P_1}\right)(G) \eeq \left( \lfp \, \phifctres{B}{P_1}{\mass{G}}\right)(G) 
    \qimplies
    \mass{\left( \lfp \, \phifct{B}{P_1}\right)(G)} \lleq \mass{G}~.
\]
\emph{$\left( \ps_r, \, \preceq \right)$ is an $\omega$-complete partial order.}
The fact that $\left( \ps_r, \, \preceq \right)$ is a partial order is immediate. It remains to show $\omega$-completeness.
For that, let $f_1 \preceq f_2 \preceq \ldots$ be an $\omega$-chain in $\ps_r$.
We have to show that $\sup_n F_n \in \ps_r$, which is the case if and only if%
\[
 \mass{\sup_n f_n} \eeq \sum_{\sigma \in \N^k} \sup_n \extract{\sigma}{f_n} \lleq r ~.
\]
Now let $g \colon \N \to \N^k$ be some bijection from $\N$ to $\N^k$.
We have 
\begin{align*}
   & \sum_{\sigma \in \N ^ k} \sup_n \extract{\sigma}{f_n} \\
   \eeq & \sum_{i = 0}^{\infty} \sup_n \extract{g(i)}{f_n} 
   \tag{series converges absolutely}\\
   \eeq & \sup_N \sum_{i=0}^{N} \sup_n \extract{g(i)}{f_n}
   \tag{rewrite infinite series as supremum of partial sums} \\
   \eeq &\sup_N \sup_n \sum_{i=0}^{N} \extract{g(i)}{f_n}
   \tag{apply monotone sequence theorem} \\
   \eeq &\sup_n \sup_N \sum_{i=0}^{N} \extract{g(i)}{f_n}
   \tag{swap suprema}
\end{align*}
Now observe that $\sup_N \sum_{i=0}^{N} \extract{g(i)}{f_n} = \mass{f_n}$, which is a monotonically increasing sequence in $n$. Moreover, since $f_n \in \ps_r$, this sequence is bounded from above by $r$. Hence, the \emph{least} upper bound $\sup_n \mass{f_n}$ of the sequence $\mass{f_n}$ is no larger than $r$, too.
This completes the proof.
\end{proof}
\subsection{Proofs of \Cref{ssec:sem_props}}
\begin{lemma}[Representation of $\sem{\while}$]
	Let $W = \WHILEDO{B}{P}$ be a \pgcl program. An alternative representation is:
	\[
	\sem{W} = \lambda G.~\sum_{i=0}^{\infty} \constrain{\varphi^i (G)}{\neg B},\quad \textnormal{where}~\varphi(G) = \sem{P}(\constrain{G}{B}).
	\]
\end{lemma}
\begin{proof}
	First we show by induction, that $\Phi_{B,P}^n(\boldsymbol{0})(G) = \sum_{i=0}^{n-1} \constrain{\varphi^i(G)}{\neg B}$. \\ \\
	\noindent
	\emph{Base case.} We have
	\[
		\Phi_{B,P}^0(\boldsymbol{0})(G) = 0 = \sum_{i=0}^{-1}\constrain{\varphi^i(G)}{\neg B}~.
	\]
	\emph{Induction step.} We have
	\begin{align*}
												&\Phi_{B,P}^{n+1}(\boldsymbol{0})(G) = \Phi_{B,P}\left(\Phi_{B,P}^{n}(\boldsymbol{0}) (G)\right)\\
												&\quad =\constrain{G}{\neg B} + \Phi_{B,P}^{n}(\boldsymbol{0})(\sem{P}\constrain{G}{B})\\
												&\quad =  \constrain{G}{\neg B} + \Phi_{B,P}^{n}(\boldsymbol{0})(\varphi(G))\\
												&\quad = \constrain{G}{\neg B} + \sum_{i=0}^{n-1}\constrain{\varphi^{i+1}}{\neg B}\\
												&\quad =\constrain{G}{\neg B} + \sum_{i=1}^{n}\constrain{\varphi^{i}}{\neg B}\\
												&\quad = \sum_{i=0}^{n}\constrain{\varphi^i(G)}{\neg B} ~.
	\end{align*}
	Overall, we thus get 
	\begin{align*}
	&\sem{W}(G) \\
	\eeq& \sup_{n \in \N} \left\lbrace \Phi_{B,P}^{n}(\boldsymbol{0})\right\rbrace(G)\\
	\eeq& \sup_{n \in \N} \left\lbrace \Phi_{B,P}^{n}(\boldsymbol{0})(G)\right\rbrace
	\tag{$\sup$ on $\ps \to \ps$ is defined point--wise}\\
	\eeq& \sup_{n \in \N} \left\lbrace\sum_{i=0}^{n} \constrain{\varphi^i (G)}{\neg B}\right\rbrace
	\tag{see above}\\
	\eeq& \sum_{i=0}^{\infty} \constrain{\varphi^i (G)}{\neg B}\tag*{\qedhere}
	\end{align*}
\end{proof}
\linearity*
\begin{proof} 
\noindent
		\textbf{Linearity of $\constrain{\cdot}{B}$.} We have
			\begin{align*}
			\constrain{a \cdot G + F}{B} &= \constrain{a \cdot \sum_{\sigma \in \N^k} \mu_\sigma X^\sigma + \sum_{\sigma \in \N^k}\nu_\sigma X^\sigma}{B}\\
			&= \constrain{\sum_{\sigma \in \N^k} \left(a \cdot \mu_\sigma + \nu_\sigma \right) X^\sigma}{B}\\
			&= \sum_{\sigma \in B} (a \cdot \mu_\sigma + \nu_\sigma) X^\sigma\\
			&= \sum_{\sigma \in B} a \cdot \mu_\sigma X^\sigma
			 \quad +\quad \sum_{\sigma \in B}
			\nu_\sigma X^\sigma\\
			&= a\cdot \sum_{\sigma \in B}
			\mu_\sigma X^\sigma \quad +\quad  \sum_{\sigma \in B} \nu_\sigma X^\sigma\\
			&= a \cdot \constrain{G}{B} + \constrain{F}{B}\qedhere
			\end{align*}
	\textbf{Linearity of $\sem{P}$}.
		By induction on the structure of $P$. First, we consider the base cases.\\ \\
		\noindent
		\emph{The case $P = \codify{skip}$.} We have
		\[
			 \sem{\codify{skip}}(r \cdot F+ G) \eeq r \cdot F + G \eeq r \cdot \sem{\codify{skip}}(F) + \sem{\codify{skip}}(G)
		\]
		\emph{The case $P = \codify{x}_i \coloneqq E $}
		\begin{align*}
			& \sem{\codify{$X_i := E$}}(r \cdot F + G) \\
			\eeq & \sum_{\sigma \in \N^k} \extract{\sigma}{r \cdot F + G} X_1^{\sigma_1}\cdots X_i^{\textnormal{eval}_\sigma(E)} \cdots X_k^{\sigma_k} \\
			\eeq &\sum_{\sigma \in \N^k}\left( r\cdot \extract{\sigma}{F} + \extract{\sigma}{G}\right) \cdot X_1^{\sigma_1}\cdots X_i^{\textnormal{eval}_\sigma(E)} \cdots X_k^{\sigma_k}
			\tag{$+$ and $\cdot$ defined coefficient--wise} \\
			\eeq & r \cdot \sum_{\sigma \in \N^k}\left(\extract{\sigma}{F}  \cdot X_1^{\sigma_1}\cdots X_i^{\textnormal{eval}_\sigma(E)} \cdots X_k^{\sigma_k} \right)
			+
			\left(\extract{\sigma}{G}  \cdot X_1^{\sigma_1}\cdots X_i^{\textnormal{eval}_\sigma(E)} \cdots X_k^{\sigma_k} \right)
			\tag{$+$ and $\cdot$ defined coefficient--wise} \\
	         \eeq &r \cdot \sum_{\sigma \in \N^k}\left(\extract{\sigma}{F}  \cdot X_1^{\sigma_1}\cdots X_i^{\textnormal{eval}_\sigma(E)} \cdots X_k^{\sigma_k} \right) \\
	         & \qquad +
	         \sum_{\sigma \in \N^k}\left(\extract{\sigma}{G}  \cdot X_1^{\sigma_1}\cdots X_i^{\textnormal{eval}_\sigma(E)} \cdots X_k^{\sigma_k} \right)
	         \tag{$+$ and $\cdot$ defined coefficient--wise} \\
	         \eeq & r \cdot \semapp{P}{F} + \semapp{P}{G}~.
		\end{align*}
		Next, we consider the induction step. \\ \\
		\noindent
		\emph{The case $P = P_1;P_2$.} We have
		\begin{align*}
			&\semapp{P_1; P_2}{r \cdot F + G} \\
			\eeq & \semapp{P_2}{\semapp{P_1}{r \cdot F + G}} \\
			\eeq & \semapp{P_2}{r \cdot \semapp{P_1}{F} + \semapp{P_1}{G}}
			\tag{I.H.\ on $P_1$} \\
			\eeq & r \cdot \semapp{P_2}{\semapp{P_1}{F}} + \semapp{P_2}{\semapp{P_1}{G}}~.
			\tag{I.H.\ on $P_2$}	
		\end{align*}
		\emph{The case $P = \ITE{B}{P_1}{P_2}$.} We have
			\begin{align*}
			& \sem{\ITE{B}{P_1}{P_2}}(r \cdot F + G)\\
			\eeq& \constrain{\sem{\codify{$P_1$}}(r \cdot F + G)}{B} + \constrain{\sem{\codify{$P_2$}}(r \cdot F + G)}{\neg B} \\
			\eeq&  \constrain{r \cdot \sem{\codify{$P_1$}}(F) + \sem{\codify{$P_1$}}(G))}{B} + \constrain{r \cdot \sem{\codify{$P_2$}}(F) + \sem{\codify{$P_2$}}(G)}{\neg B} 
			\tag{I.H.\ on $P_1$ and $P_2$}\\
			\eeq & r \cdot \big(\constrain{\sem{P_1}(F)}{B} + \constrain{\sem{P_2}(F)}{\neg B}\big) + \constrain{\sem{P_1}(G)}{B} + \constrain{\sem{P_2}(G)}{\neg B} 
			\tag{linearity of $\constrain{\cdot}{B}$ and $\constrain{\cdot}{\neg B}$}\\
			\eeq& r \cdot \sem{\ITE{B}{P_1}{P_2}}(F) + \sem{\ITE{B}{P_1}{P_2}}(G)\\
			\end{align*}
			\emph{The case $P = \PCHOICE{P_1}{p}{P_2}$.}
			\begin{align*}
			&\sem{\PCHOICE{P_1}{p}{P_2}}(r \cdot F + G) \\
			\eeq& p \cdot \sem{P_1}(r\cdot F + G) + (1-p)\cdot \sem{P_2}(r\cdot F + G)\\
			\eeq& p \cdot \left(r \cdot \sem{P_1}(F) + \sem{P_1}(G) \right) + (1-p)\cdot \left( r\cdot \sem{P_2}(F) + \sem{P_2}(G)\right)
			\tag{I.H.\ on $P_1$ and $P_2$}\\
			\eeq &  r \cdot \left( p \cdot \sem{P_1}(F) + (1-p) \cdot \sem{P_2}(F)\right) + p \cdot \sem{P_1}(G) + (1-p) \cdot \sem{P_2}(G) 
			\tag{reorder terms}\\
			\eeq& r \cdot \sem{\PCHOICE{P_1}{p}{P_2}}(F) + \sem{\PCHOICE{P_1}{p}{P_2}}\\
			\end{align*}
		\emph{The case $P= \WHILEDO{B}{P_1}$.}
			\begin{align*}
			&\sem{\WHILEDO{B}{P_1}}(r \cdot F + G) \\
			\eeq& \sup_{n \in \N} \left\lbrace \Phi_{B, P_1}^{n}(\boldsymbol{0})\right\rbrace(r\cdot F + G)\\
			\eeq& \sup_{n \in \N} \left\lbrace\Phi_{B,P_1}^n(\boldsymbol{0})(r \cdot F + G)\right\rbrace
			\tag{$\sup$ on $\ps \to \ps$ defined point--wise}
			\\
			\eeq& \sup_{n \in \N} \left\lbrace r\cdot \Phi_{B,P_1}^n(\boldsymbol{0})(F) \pplus \Phi_{B,P_1}^n(\boldsymbol{0})(G)\right\rbrace\tag{by straightforward induction on $n$ using I.H.\ on $P_1$}\\
			\eeq& r\cdot \sup_{n \in \N} \left\lbrace \Phi_{B,P_1}^n(\boldsymbol{0})(F)\right\rbrace + \sup_{n \in \N} \left\lbrace\Phi_{B,P_1}^n(\boldsymbol{0})(G)\right\rbrace
			\tag{apply monotone sequence theorem coefficient--wise}\\
			\eeq& r \cdot \sem{\WHILEDO{B}{P_1}}(F) \pplus \sem{\WHILEDO{B}{P_1}}(G)
		\end{align*}
	\textbf{Linearity of $\Phi_{B,P}(f)$ for linear $f$.}
		\begin{align*}
			\Phi_{B,P}(f)\left(\sum_{\sigma \in \N^k} \mu_\sigma X^\sigma\right) &= \constrain{\sum_{\sigma \in \N^k} \mu_\sigma X^\sigma}{\neg B} + f\left(\sem{P}\left(\constrain{\sum_{\sigma \in \N^k} \mu_\sigma X^\sigma}{B}\right)\right)\\
			&= \constrain{\sum_{\sigma \in \N^k} \mu_\sigma X^\sigma}{\neg B} + f \left(\sum_{\sigma \in \N^k} \mu_\sigma \sem{P}\left(\constrain{X^\sigma}{B}\right)\right)\tag{1. \& 2.}\\
			&= \constrain{\sum_{\sigma \in \N^k} \mu_\sigma X^\sigma}{\neg B} + \sum_{\sigma \in \N^k} \mu_\sigma \cdot f \left( \sem{P}\left(\constrain{X^\sigma}{B}\right)\right)\tag{$f$ lin.}\\
			&= \sum_{\sigma \in \N^k} \mu_\sigma \constrain{X^\sigma}{\neg B} + \mu_\sigma \cdot f \left( \sem{P}\left(\constrain{X^\sigma}{B}\right)\right)\\
			&= \sum_{\sigma \in \N^k} \mu_\sigma \cdot \big( \constrain{X^\sigma}{\neg B} + f \left( \sem{P}\left(\constrain{X^\sigma}{B}\right)\right)\big)\\
			&= \sum_{\sigma \in \N^k} \mu_\sigma \cdot \Phi_{B,P}(f)(X^\sigma) \qedhere
		\end{align*}
\end{proof}
\unrolling*
\begin{proof}
	Let $W, W'$ be as described in \Cref{lem:loop_unrolling}.
\begin{align*}
\sem{W}(G) &= \left(\lfp \Phi_{B,P}\right) (G)\\
				  &= \Phi_{B,P}\left( \lfp \Phi_{B,P}\right) (G)\\
				  &= \constrain{G}{\neg B} + \left(\lfp \Phi_{B,P}\right)\left(\sem{P}\big(\constrain{G}{B}\big)\right)\\
				  &= \sem{\ITE{B}{P;W}{\pskip}}(G)\\
				  &= \sem{W'}(G)\qedhere
\end{align*}
\end{proof}

\subsection{Proofs of \Cref{ssec:kozen_relation}}
\label{app:kozen}
\begin{lemma}
	The mapping $\tau$ is a bijection. The inverse $\tau^{-1}$ of $\tau$ is given by
	\[ \tau^{-1}\colon \mathcal{M} \to \ps, \quad \mu \mmapsto \sum_{\sigma \in \N^k} \mu\left(\{\sigma\}\right)\cdot \bvec{X}^\sigma\]
\end{lemma}
\begin{proof}
	We show this by showing $\tau^{-1} \circ \tau = \textup{id}$ and $\tau \circ \tau^{-1} = \textup{id}$.
	\begin{align*}
	\tau^{-1} \circ \tau\left(\sum_{\sigma  \in \N^k} \alpha_\sigma \bvec{X}^\sigma\right) &= \tau^{-1}\left(\lambda N\boldsymbol{.}~ \sum_{\sigma \in N} \alpha_\sigma\right)\\
	&= \sum_{\sigma  \in \N^k} \sum_{s \in \{\sigma\}} \alpha_\sigma \cdot \bvec{X}^\sigma = \sum_{\sigma  \in \N^k} \alpha_\sigma\bvec{X}^\sigma\\
	\tau \circ \tau^{-1} \left(\mu\right) &= \tau\left(\sum_{\sigma  \in \N^k} \mu(\{\sigma\}) \cdot \bvec{X}^\sigma\right) = \lambda N\boldsymbol{.}~ \sum_{\sigma \in N}\mu(\{\sigma\}) = \mu(N) = \mu
	\end{align*}
\end{proof}

\begin{lemma}
	The mappings $\tau$ and $\tau^{-1}$ are monotone linear maps.
\end{lemma}
\begin{proof}
	First, we show that $\tau^{-1}$ is linear (and hence $\tau$, due to bijectivity):
	\begin{align*}
	\tau^{-1}(\mu + \nu) &= \sum_{\sigma  \in \N^k} (\mu + \nu) (\{\sigma\})\cdot \bvec{X}^\sigma\\
	&= 	\sum_{\sigma  \in \N^k} (\mu(\{\sigma\}) + \nu(\{\sigma\})) \cdot \bvec{X}^\sigma\tag{as $\mathcal{M}$ forms a vector space with standard +}\\
	&= \sum_{\sigma  \in \N^k} \left(\mu(\{\sigma\})\cdot \bvec{X}^\sigma + \nu(\{\sigma\}) \cdot \bvec{X}^\sigma\right)\\
	&=  \left(\sum_{\sigma  \in \N^k} \mu(\{\sigma\})\cdot \bvec{X}^\sigma\right) +  \left(\sum_{\sigma  \in \N^k} \nu(\{\sigma\})\cdot \bvec{X}^\sigma\right) = \tau^{-1}(\mu) + \tau^{-1}(\nu) 
	\end{align*}
	
	Second, we show that $\tau$ is monotone:
	\begin{align*}
	\text{Assume} ~ G_\mu \sqsubseteq G_{\mu'} &.\\
	\tau(G_\mu) &= \tau\left(\sum_{\sigma  \in \N^k} \mu(\{\sigma\})\cdot\bvec{X}^\sigma\right) = \lambda S\boldsymbol{.}\; \sum_{\sigma  \in S} \mu(\{\sigma\})\\
	&\leq \lambda S\boldsymbol{.}\; \sum_{\sigma  \in S} \mu'(\{\sigma\}) \tag{as $\mu(\{\sigma\}) \leq \mu'(\{\sigma\})$ per definition of $\sqsubseteq$}\\
	&= \tau\left(\sum_{\sigma  \in \N^k} \mu'(\{\sigma\})\cdot\bvec{X}^\sigma\right) = \tau(G_{\mu'}) 
	\end{align*}
	
	Third, we show that $\tau^{-1}$ is monotone:
	\begin{align*}
	\text{Assume} ~ \mu \sqsubseteq \mu' &.\\
	\tau^{-1}(\mu) &= \sum_{\sigma  \in \N^k} \mu(\{\sigma\})\cdot\bvec{X}^\sigma\\
	&\sqsubseteq \sum_{\sigma  \in \N^k} \mu'(\{\sigma\})\cdot\bvec{X}^\sigma \tag{as $\mu(\{\sigma\}) \leq \mu'(\{\sigma\})$ per definition of $\sqsubseteq$}\\
	&= \tau^{-1}(\mu') \tag*{\qedhere}
	\end{align*}
\end{proof}

\begin{lemma}
	Let $f\colon (P, \leq) \to (Q, \leq)$ be a monotone isomorphism for any partially ordered sets $P$ and $Q$. Then, \[f^\ast\colon \textup{Hom}(P,P) \to \textup{Hom}(Q,Q),\quad \phi \mmapsto f \circ \varphi \circ f^{-1}\] is also a monotone isomorphism.
\end{lemma}
\begin{proof}
	Let $f$ be such a monotone isomorphism, and $f^\ast$ the corresponding lifting.
	First, we note that $f^\ast$ is also bijective.
	Its inverse is given by $(f^\ast)^{-1} = (f^{-1})^\ast$.
	Second, $f^\ast$ is monotone, as shown in the following calculation.
	
	\begin{align*}
	f \leq g &\iimplies \forall x. \quad f(x) \lleq g(x)\\
	&\iimplies \forall x.\quad \tau \circ f \left(\tau^{-1}\circ \tau (x)\right) \lleq \tau \circ g\left(\tau^{-1}\circ \tau (x)\right)\\
	&\iimplies \forall x. \quad \tau^\ast \circ f \left(\tau(x)\right) \lleq \tau^\ast \circ g\left(\tau (x)\right)\\
	&\iimplies \forall y. \quad \tau^\ast \circ f(y) \lleq \tau^\ast \circ g(y)\\
	&\iimplies \tau^\ast(f) \lleq \tau^\ast (g)
	\end{align*}
\end{proof}

\begin{lemma}\label{lem:helper}
	Let $P,Q$ be complete lattices, and $\tau$ a monotone isomorphism. Also let \textsf{\textup{lfp}} be the least fixed point operator. Then the following diagram commutes.
	\begin{center}
		\begin{tikzcd}
		\textup{Hom}(P,P) \ar[rr, "\tau^\ast"] \ar[dd, swap, "\lfp"] && \textup{Hom}(Q,Q) \ar[dd, "\lfp"]\\
		\\
		P \ar[rr, "\tau", swap] && Q
		\end{tikzcd}
	\end{center}
\end{lemma}
\begin{proof}
	Let $\varphi \in \textup{Hom}(P,P)$ be arbitrary.
	\begin{align*}
	\lfp \varphi &\eeq \inf \left\{p \bigmid \varphi(p) \eeq p\right\}\\
	\tau\left(\lfp \varphi\right) &\eeq \tau\left(\inf \left\{p \bigmid \varphi(p) \eeq p\right\}\right)\\
	&\eeq \inf \left\{\tau(p) \bigmid \varphi (p) \eeq p\right\}\\
	&\eeq \inf \left\{\tau(p)\bigmid\varphi (\tau^{-1}\circ\tau (p)) \eeq \tau^{-1}\circ\tau(p)\right\}\\
	&\eeq \inf \left\{\tau(p)\bigmid \tau \circ \varphi (\tau^{-1}\circ\tau (p)) \eeq \tau(p)\right\}\\
	&\eeq \inf \left\{q\bigmid \tau \circ \varphi(\tau^{-1} (q)) \eeq q\right\}\\
	&\eeq \inf \left\{q\bigmid \tau^\ast (\varphi) (q) \eeq q\right\}\\
	&\eeq  \lfp \tau^\ast(\varphi)\tag*{\qedhere}
	\end{align*}
\end{proof}

\begin{definition}
	Let $\mathfrak{T}$ be the program translation from \pgcl to a modified Kozen syntax, defined inductively:
	\begin{align*}
	\mathfrak{T}(\pskip) &\eeq \pskip \\
	\mathfrak{T}(\ASSIGN{x_i}{E}) &\eeq \ASSIGN{x_i}{f_E(x_1,\ldots,x_k)}\\
	\mathfrak{T}(\PCHOICE{P}{p}{Q}) &\eeq \PCHOICE{\mathfrak{T}(P)}{p}{\mathfrak{T}(Q)}\\
	\mathfrak{T}(\compose{P}{Q}) &\eeq \mathfrak{T}(P);\mathfrak{T}(Q)\\
	\mathfrak{T}(\ITE{B}{P}{Q}) &\eeq \textup{if } B \textup{ then } \mathfrak{T}(P) \textup{ else } \mathfrak{T}(Q) \textup{ fi}\\
	\mathfrak{T}(\WHILEDO{B}{P}) &\eeq \textup{while } B \textup{ do } \mathfrak{T}(P) \textup{ od}~,
	\end{align*}
	where $p$ is a probability, $k = \abs{\var(P)}$, $B$ is a Boolean expression and $P, Q$ are \pgcl programs. The extended construct \pskip as well as $\PCHOICE{P}{p}{Q}$ is only syntactic sugar and can be simulated by the original Kozen semantics. The intended semantics of these constructs are \begin{align*}
	[\pskip] &= \textup{id}\\
	\text{and} \qquad \left[\PCHOICE{P}{p}{Q}\right] &= p \cdot \mathfrak{T}(P) + (1-p) \cdot \mathfrak{T}(Q).
	\end{align*}
\end{definition}

\begin{lemma}
	For all guards $B$, the following identity holds: $\text{e}_{B} \circ \tau \eeq \tau \circ \constrain{\cdot}{B}$.
\end{lemma}
\begin{proof}
	For all $G_\mu = \sum_{\sigma  \in \N^k}\mu\left(\{\sigma\}\right) \cdot \bvec{X}^\sigma\in \ps$:
	\begin{align*}
	\text{e}_B \circ \tau (G_\mu) &= \text{e}_B (\mu)\\
	&= \lambda S\boldsymbol{.}~ \mu(S \cap B)\\
	\tau \circ \constrain{G_\mu}{B} &= \tau\left(\sum_{\sigma \in B} \mu\left(\{\sigma\}\right)\cdot \bvec{X}^\sigma + \sum_{\sigma \not\in B} 0\cdot \bvec{X}^\sigma\right)\\
	&= \lambda S\boldsymbol{.}~ \mu(S \cap B)
	\end{align*}
	$\implies\qquad \forall G_\mu \in \ps.~~ \text{e}_B \circ \tau (G_\mu) \eeq \tau \circ \constrain{G_\mu}{B}$
\end{proof}

\kozenrel*
\begin{proof}
	The proof is done via induction on the program structure. We omit the loop-free cases, as they are straightforward.
	
	By definition, $\mathfrak{T}(\WHILEDO{B}{P}) \eeq \codify{while}~ B~\codify{do}~P~\codify{od}$. Hence, the corresponding Kozen semantics is equal to $\lfp T_{B,P}$, where \[T\colon (\mathcal{M} \to \mathcal{M}) \to (\mathcal{M} \to \mathcal{M}), \quad S \mmapsto \text{e}_{\bar B} + (S \circ P \circ \text{e}_{B})~.\]
	First, we show that $\tau^{-\ast} \circ T_{B,P} \circ \tau^{\ast} \eeq \Phi_{B,P}$, where $\tau^\ast$ is the canonical lifting of $\tau$, i.e., $\tau^\ast (\mathfrak{S}) \eeq \tau \circ \mathfrak{S} \circ \tau^{-1}$ for all $\mathfrak{S} \in (\ps \to \ps)$.
		\begin{align*}
		\left[\tau^{-\ast} \circ T_{B,P} \circ \tau^{\ast} \right](\mathfrak{S})&= \tau^{-\ast} \circ T_{B,P} \circ \tau \circ \mathfrak{S} \circ \tau^{-1}\\
		&= \tau^{-\ast}\left(\text{e}_{\bar B} \pplus \tau \circ \mathfrak{S} \circ \tau^{-1} \circ P\circ \text{e}_{B}\right)\\
		&= \tau^{-1} \circ \text{e}_{\bar B} \circ \tau \pplus \tau^{-1} \circ \tau \circ \mathfrak{S} \circ \tau^{-1} \circ P \circ \text{e}_{B} \circ \tau\\
		&= \tau^{-1} \circ \text{e}_{\bar B} \circ \tau \pplus \mathfrak{S} \circ \tau^{-1} \circ P \circ \text{e}_{B} \circ \tau\\
		&= \tau^{-1} \circ \tau \circ \constrain{\cdot }{\bar B}  \pplus \mathfrak{S} \circ \tau^{-1} \circ P \circ \tau \circ \constrain{\cdot}{B}\\
		&= \constrain{\cdot}{\bar B} \pplus \mathfrak{S} \circ \tau^{-1} \circ \tau \circ \sem{P} \circ \constrain{\cdot}{B} \tag{Using I.H. on $P \circ \tau$}\\
		&= \constrain{\cdot}{\bar B} \pplus \mathfrak{S} \circ \sem{P} \circ \constrain{\cdot}{B}\\
		&= \Phi_{B,P}(\mathfrak{S})
		\end{align*}
	Having this equality at hand, we can easily proof the correspondence of our \while semantics to the one defined by Kozen in the following manner:
	\begin{align*}
		&\quad\tau \circ \sem{\WHILEDO{B}{P}} \eeq \mathfrak{T}(\WHILEDO{B}{P}) \circ \tau\\
		\Leftrightarrow&\quad \tau \circ \lfp \Phi_{B,P} \eeq \lfp T_{B,P} \circ \tau\\
		\Leftrightarrow&\quad \lfp \Phi_{B,P} \eeq \tau^{-1} \circ \lfp T_{B,P} \circ \tau\\
		\Leftrightarrow&\quad \lfp \Phi_{B,P} \eeq \tau^{-\ast} (\lfp T_{B,P})\tag{Definition of $\tau^\ast$}\\
		\Leftrightarrow&\quad \lfp \Phi_{B,P} \eeq \lfp \left(\tau^{-\ast} \circ T_{B,P} \circ \tau^{\ast}\right) \tag{cf. \cref{lem:helper}}\\
		\Leftrightarrow&\quad \lfp \Phi_{B,P} \eeq \lfp \Phi_{B,P}\tag*{\qedhere}
	\end{align*}

\end{proof}
\section{Proofs of \Cref{sec:overapprox}}
\label{app:over_approx}

\overapprox*%
\begin{proof}
	Instance of Park's Lemma~\cite{?}.
\end{proof}

\superinvs*
\begin{proof}
	Let $G \in \ps$ be arbitrary.
	\begin{align*}
	\Phi_{B,P}(\hat{f})(G) ~&=~ \sum_{\sigma  \in \N^k} \extract{\sigma}{G} \Phi_{B,P}(\hat{f})(\bvec{X}^\sigma) \tag{By \Cref{lem:lin_funcs}}\\
	~&\sqsubseteq~ \sum_{\sigma  \in \N^k} \extract{\sigma}{G} f(\bvec{X}^\sigma) \tag{By assumption}\\
	~&=~ \hat{f}(G)\\
	\implies\quad\sem{W}~&\sqsubseteq~ \hat{f} \tag{By \cref{thm:overapprox}}
	\end{align*}
	
\end{proof}

\begin{proof}[Proof of \Cref{ex:rw}]
		\begin{align*}
			\Phi_{B,P}\left(\hat{f}\right)\left(X^iC^j\right) ~&=~ \left(\constrain{X^iC^j}{i=0} \pplus \hat{f}\left(\frac{1}{2}\constrain{X^iC^j}{i > 0} \cdot \frac{C}{X} \pplus \frac{1}{2}\constrain{X^iC^j}{i > 0} \cdot XC\right)\right)\\
			\textnormal{case } i=0\colon&\Rightarrow ~(C^j+\hat{f}(0)) \eeq C^j \eeq f(X^0C^j)\\
			\textnormal{case } i>0\colon&\Rightarrow ~ \frac{C}{2}\left(\hat{f}(X^{i-1}C^j) \pplus \hat{f}\left(X^{i+1}C^j\right)\right)\\
			&=~ C^j\cdot\begin{cases}
			\frac{1}{1-C^2},&i~ \textnormal{even}\\
			\frac{C}{1-C^2},&i~ \textnormal{odd}\\
			\end{cases} \eeq f(X^iC^j)\\
			\implies&~ \Phi_{B,P}\left(\hat{f}\right)(X^iC^j) \ssqsubseteq f(X^iC^j).
		\end{align*}
		Thus $\hat{f}$ is a superinvariant.
		\begin{align*}
		\Phi_{B,P}(\hat{h})(X^iC^j) ~&=~ \constrain{X^iC^j}{i=0} \pplus \hat{h}\left(\frac{1}{2}\constrain{X^iC^j}{i > 0} \cdot \frac{C}{X} \pplus \frac{1}{2}\constrain{X^iC^j}{i > 0} \cdot XC\right)\\
		\textnormal{case } i=0\colon&\Rightarrow ~(C^j+\hat{h}(0)) \eeq  1 \eeq h(X^0C^j)\\
		\textnormal{case } i>0\colon&\Rightarrow ~ \frac{C}{2}\left(\hat{h}\left(X^{i-1}C^j\right)\pplus \hat{h}\left(X^{i+1}C^j\right)\right)\\
		&=~ \frac{C^{j+1}}{2} \cdot \left(\left(\frac{1-\sqrt{1-C^2}}{C}\right)^{i-1} \pplus \left(\frac{1-\sqrt{1-C^2}}{C}\right)^{i+1}\right)\\
		&=~ C^j\cdot \left(\frac{1-\sqrt{1-C^2}}{C}\right)^i \eeq h(X^iC^j)\\
		\implies&~ \Phi_{B,P}(\hat{h})(X^iC^j) \eeq h(X^iC^j)
		\end{align*}
		Thus $\hat{f}$ is a superinvariant.
\end{proof}

\noindent\emph{Verification Python Script}
\definecolor{commentsColor}{rgb}{0.497495, 0.497587, 0.497464}
\definecolor{keywordsColor}{rgb}{0.000000, 0.000000, 0.635294}
\definecolor{stringColor}{rgb}{0.558215, 0.000000, 0.135316}
\begin{lstlisting}[
	caption=Python program checking the invariants,
	captionpos=b,
	frame=l,
	framerule=1.5pt,
	framesep=2pt,
	label=prog:sympy,
	language=Python,
	breaklines=true,
	backgroundcolor = \color{lightgray!50!white},
	basicstyle=\scriptsize,
	commentstyle=\color{commentsColor}\textit,
	keywordstyle=\color{keywordsColor}\bfseries,
	stringstyle=\color{stringColor},
	showstringspaces=false,
	tabsize=2
	]
from sympy import *
init_printing()

x, c = symbols('x,c')
i, j = symbols('i,j', integer=True)

#define the higher order transformer
def Phi(f):
	return c/2 * (f.subs(i,i-1) + f.subs(i,i+1))

def compute_difference(f):
	return (Phi(f) - f).simplify()

#define Polynomial verifyer
def verify_poly(poly):
	print("Check coefficients for non-positivity:")
	for coeff in poly.coeffs():
    	if coeff > 0:
        	return False
     return True

# actual verification method
def verify(f):
	print("Check invariant:")
	pprint(f)
	result = compute_difference(f)
	if result.is_zero:
		print("Invariant is a fixpoint!")
		return True
	else:
		print("Invariant is not a fixpoint - check if remainder is Poly")
		try:
			return verify_poly(Poly(result))
		except PolificationFailed:
			print("Invariant is not a Poly!")
			return False
		except:
			print("Unexpected Error")
			raise             

# define the loop invariant guess (i != 0) case
f = c**j * ((c / (1-c**2)) * (i % 2) + (1/(1-c**2)) * ((i+1) % 2))

# Second invariant:
h = c**j * ( ( 1 - sqrt(1 - c**2) ) / c )**i

print("Invariant verified" if verify(f) else "Unknown")
print("Invariant verified" if verify(h) else "Unknown")
\end{lstlisting}

\begin{proof}[Proof of \Cref{ex:non_ast}]
		\begin{align*}
			\Phi_{B,P}(\hat{f})(X^i) ~&=~ \constrain{X^i}{i = 0} \pplus \frac{1}{i} \cdot \hat{f}\left(\constrain{X^i}{i > 0} \cdot \frac{1}{X}\right) \pplus \left(1 - \frac{1}{i}\right) \cdot \hat{f}\left(\constrain{X^i}{i > 0} \cdot X\right)\\
			\textnormal{case } i=0\colon~&\Rightarrow~ 1 \pplus \infty \cdot \hat{f}(0) \pplus -\infty \cdot \hat{f}(0)\\
			~&=~\; 1 \pplus \infty \cdot 0 \pplus -\infty \cdot 0 \eeq 1 \eeq f\left(X^i\right)\\
			\textnormal{case } i>0\colon~&\Rightarrow~ 0 \pplus \frac{1}{i} \cdot \hat{f}\left(X^{i-1}\right) \pplus \left(1-\frac{1}{i}\right) \cdot \hat{f}\left(X^{i+1}\right)\\
			~&=~ \frac{1}{i}\cdot \left(1-\frac{1}{e}\cdot \sum_{n=0}^{i-3}\frac{1}{n!}\right) \pplus \left(1- \frac{1}{i}\right) \cdot \left(1-\frac{1}{e}\cdot \sum_{n=0}^{i-1}\frac{1}{n!}\right)\\
			~&=~\frac{1}{i} \mminus \left(\frac{1}{ei}\cdot \sum_{n=0}^{i-3}\frac{1}{n!}\right) \pplus \left(1-\frac{1}{e}\cdot \sum_{n=0}^{i-1}\frac{1}{n!}\right) \mminus  \frac{1}{i} \pplus \left(\frac{1}{ei}\cdot \sum_{n=0}^{i-1}\frac{1}{n!}\right)\\
			~&=~  \left(1-\frac{1}{e}\cdot \sum_{n=0}^{i-1}\frac{1}{n!}\right) \pplus \frac{1}{ei}\cdot \left(\frac{1}{(i-2)!} + \frac{1}{(i-1)!}\right)\\
			~&=~ \left(1-\frac{1}{e}\cdot \sum_{n=0}^{i-1}\frac{1}{n!}\right) \pplus \frac{1}{e(i-1)!} \eeq \left(1-\frac{1}{e}\cdot \sum_{n=0}^{i-2}\frac{1}{n!}\right)\\
			~&=~ f\left(X^i\right)\\
			\implies~&~ \Phi_{B,P}(\hat{f})(X^i) \eeq f(X^i)
		\end{align*}
		
		\noindent\emph{Mathematica input query:} 
		\begin{align*}
		\texttt{Input:}&\qquad\frac{1}{k} \cdot  \left(1-\frac{1}{e} \cdot \sum_{n = 0}^{k-3}\frac{1}{n!}\right) +  (1-\frac{1}{k}) \cdot \left(1-\frac{1}{e} \cdot  \sum_{n = 0}^{k-1}\frac{1}{n!}\right) - \left(1-\frac{1}{e} \cdot \sum_{n=0}^{k-2}\frac{1}{n!}\right)\\
		\texttt{Output:}&\qquad 0\qedhere
		\end{align*}
\end{proof}

\end{document}